\documentclass[11pt]{amsart}
\usepackage{amsfonts,amssymb,graphicx,mystyle,color,amsmath}
\usepackage[round]{natbib}
\usepackage{array}
\usepackage{tikz,istgame}
\usepackage{enumitem}
\usepackage{comment}
\usepackage{setspace}
\usepackage[colorlinks,citecolor={blue},linkcolor=blue, runcolor={blue}]{hyperref} 
\usepackage{bbm}
    
\def\a{\alpha} \def\b{\beta} \def\d{\delta} \def\e{\varepsilon}  \def\g{\gamma}  \def\s{\sigma}  
\def\D{\Delta} \def\G{\Gamma}  
\def\S{\Sigma}
     \def\N{\mathcal N}

\def\Re{\mathbb{R}}

\def\rho{\varrho}

\def\a{\alpha}
\def\d{\delta}

\begin{document}\openup 1\jot
\title[Index and Robustness in Mixed Equilibria]{Index and Robustness of Mixed Equilibria: an algebraic approach \\}
\author[L. Pahl]{Lucas Pahl}
\thanks{I am grateful to Srihari Govindan, Paulo Barelli,  Robert B. Wilson, Yehuda ``John'' Levy, Julien Fixary, Rida Laraki, Mich Tvede, Josef Hofbauer, Klaus Ritzberger, Miklos Pinter, Michael Greinecker, and Irem Portakal for comments and stimulating conversations. I thank participants of the seminars at the School of Economics of the University of Sheffield, SAET-2026, and EWM-2026}
\address{School of Economics, University of Sheffield, Elmﬁeld Building, Northumberland Rd, Sheffield S10 2TU}
\date{\today}

\begin{abstract}We present two results regarding the relation of the index of equilibria and the index of fixed points, and the index of equilibria and payoff-robustness: the main result is that any integer can be the index of an isolated completely mixed equilibrium of a finite game. A second result is that in a particular class of isolated completely mixed equilibria, called monogenic, the index can be $0$, $+1$ or $-1$ only. In this class non-zero index is equivalent to payoff-robustness. We also discuss extensions of the method of computation to extensive-form games, and cases where the equilibria might be located on the boundary of the strategy set. These results follow from a new method of computing indices of equilibria, which we also develop in this paper, that is based on algebraic properties of the system of polynomials describing completely mixed equilibria.
\end{abstract}

\maketitle

\section{Introduction}

Fixed-point theory lies at the core of game theory and economics, providing
useful tools and concepts for investigating the existence and robustness of
equilibria. Its contributions to the theory of equilibrium refinements and
equilibrium computation in games are fruitful and varied (cf. \cite{KR1994}, \cite{GPS1993}). A central tool in
fixed-point theory for checking the robustness of fixed points is the
fixed-point index. Roughly, the fixed-point index is an integer assigned to a
compact set of fixed points of a map that captures whether the set is robust to
perturbations of the map. Here, robustness means that, for every sufficiently
small perturbation of the map, there are fixed points close to the original
set. Whenever the index is non-zero, robustness obtains; whenever it is zero,
it does not  (cf. \cite{BO1953}).

Analogously, game theory has the central notion of robustness to payoff
perturbations: a set of equilibria is robust if, for every sufficiently small
payoff perturbation, there are equilibria close to the set. Robustness to
payoff perturbations is central to the theory of refinements, since it implies
virtually all the properties imposed by classical solution concepts\footnote{Normal-form, payoff-robust equilibria are perfect. Coupled with invariance to addition to duplicate strategies, payoff-robustness implies properness, and several other refinement properties (see \cite{KM1986}).} as well as underlies the formulation of modern concepts in refinement theory (cf. \cite{GW2005,RWW2025,PP2025,GW2025}).

As equilibria of games can be described as fixed points, it is natural to ask
how these two notions of robustness are related. For general fixed-point
problems, O'Neill's theorem (Theorem 5.2 in \cite{BO1953}) tells us how the set of fixed points around a
fixed-point component varies when small perturbations of the fixed-point map
are introduced. The theorem reads as follows: let $m>1$ and $f:U\to\Re^m$ be a continuous map, where $U \subset \Re^m$ is an open
neighborhood with a single compact and connected set $C$ of fixed points of
$f$. Choose any finite collection of points $\{x_i\}_i\subset C$ and, for each
$i$, an integer $k_i$ such that $\sum_i k_i $= index of $C$. Then there exists a map $h$, arbitrarily close to $f$, which has precisely the
points $x_i$ as fixed points, each with index $k_i$ with respect to $h$.\footnote{If $m=1$, then the indices of fixed points are constrained to be only $+1$, $-1$ or $0$.}

Perhaps surprisingly, a version of O'Neill's theorem also holds for games
(Theorem 2.4 in \cite{GLP2025}). Roughly, given a finite game $G$ and a component of equilibria
$C$, one can obtain a statement analogous to the one above: there exists a neighborhood $U$ of $C$ in the mixed strategy set of the players such that for any finite set of completely mixed profiles
 $\{x_i\}_i \subset U$, integers $k_i \in\{-1,+1\}$ such that $\sum_i k_i $ = index of $C$, there exists an arbitrarily small payoff perturbation $H$ of $G$ with precisely each $x_i$ as an equilibrium of $H$ in $U$, each $x_i$ with index $k_i$.\footnote{The formal statement of this theorem requires passing to an equivalent $\bar{G}$ game to $G$ (obtained from adding finitely many duplicate strategies to $G$) in order to generate the required set of equilibria $\{x_i\}_{i}$. In case $C$ intersects the interior of the mixed strategy set, then an $x_i$ can also be chosen in $C$.}

Like its general fixed-point analogue, the game-theoretic version of O'Neill's
theorem therefore shows that the index of equilibria is a fundamental
obstruction to producing an arbitrary set of equilibria around a component
under small payoff perturbations. There is, however, a noticeable gap
between the two results: the points chosen as equilibria of a perturbed game
can have indices $+1$ and $-1$ only, rather than arbitrary integers. It
is unclear whether this restriction is a consequence of the construction
undertaken in \citet{GLP2025}, or, more generally, of the game-theoretic
environment.

There are good reasons supporting both possibilities. For every integer $q$,
there exists a two-player game with a nondegenerate component of equilibria
having index $q$ (cf. Proposition 3.1 in \cite{GSS2004}). However, the construction produces a
component that is neither completely mixed nor a singleton. If the component
were a completely mixed singleton, it would necessarily have index $+1$ or
$-1$, as a consequence of the two-player environment. It is
therefore unclear whether equilibria arising from payoff perturbations of a
game in the sense of O'Neill's theorem for games could have indices other than
$+1$ and $-1$. More fundamentally, it is not even known what are the possible indices that completely mixed equilibria of finite games can take. It is known that in fixed point theory isolated fixed points can have any indices. Is the analogous result true in games? 

The main contribution of this paper is to answer this question, thereby furthering
our understanding of the fundamental similarities and differences between
general fixed-point problems and those induced by games. We prove that every
integer is the index of an isolated completely mixed equilibrium of a
three-player finite game (Theorem \ref{finalthm}). This suggests that O'Neill's theorem for games could,
in principle, be extended to encompass a richer collection of indices than
$+1$ and $-1$.

In pursuing this question, we developed a new and useful method for index
computation that allows us to refine existing results concerning the
relationship between robustness and non-zero index. It is known that
equilibria of non-zero index are payoff-robust, but the converse is false in
general (cf. \cite{HH2002}). We obtain a converse of this result (Theorem \ref{mainthm1}). If $\s^\ast$ is an isolated and completely mixed equilibrium, with the polynomial
system describing the completely mixed equilibria of the game having a Jacobian
of corank one at $\s^\ast$, then the equilibrium has index zero if and only if
it is not payoff-robust. In other words, in this case, non-zero index is
equivalent to payoff robustness. Admittedly, the situations under which the converse holds are rather restrictive: they require the equilibrium to be isolated and completely mixed (we discuss extensions in section \ref{SEC4}). However, they include non-generic cases on which payoff-robustness is very hard to check for. This, coupled with our formula for index computation, turns the index into a sharp tool to verify for payoff-robustness. 

A relevant implication of this new method for game theory is that it is
perturbation-free. In generic cases, finite games possess only regular
equilibria, whose indices are $+1$ and $-1$ and are relatively simple to
compute. The procedure amounts to computing the sign of the
Jacobian determinant of the polynomial system describing the equilibrium. In
two-player games, there is an even simpler formula due to \citet{S1974}.

In non-generic cases, however, these formulas do not apply, and computation
becomes difficult. One needs to perturb the payoffs to obtain a generic game,
using a sufficiently small perturbation so that the index is unaffected, and
then compute the regular equilibria that arise near the non-regular
equilibrium. Summing the indices of these regular equilibria yields the index
of the non-regular equilibrium. There is no general answer to how small the
perturbation must be or how all the nearby regular equilibria should be found.
Our proposed method therefore provides a
perturbation-free way of checking the robustness of equilibria even in
singular cases.

We use a motivating example throughout the paper to illustrate our results and
showcase the application of the method. The method is based
on algebraic invariants of the polynomial system describing the completely
mixed equilibria of a finite game and is derived from the formula for
computing the topological degree of isolated zeros of smooth real maps
developed by \citet{EL1977}. Using the motivating example, we show how the
index of a non-regular isolated completely mixed equilibrium can be computed
without resorting to arbitrary perturbations. 

This paper is organized as follows: section \ref{SEC2} comprises preliminary facts on algebra and topology that we use throughout this paper: we present a number of  definitions in basic abstract algebra and algebraic geometry that are necessary for the paper.  We also prove the main index formula that will be used throughout the paper (Proposition \ref{auxprop}). Section \ref{SEC3} presents the results concerning the \textit{monogenic class} of games and its relevance for payoff-robustness (Theorem \ref{mainthm1}). We also present a detailed application of the index formula for computation of the index of an equilibrium in an example (Example \ref{example}), and the result stating that isolated completely mixed equilibria can have any index (Theorem \ref{finalthm}). In Section \ref{SEC4} we discuss extensions of the results to extensive-form games, to the case of non-degenerate components of equilibria, and equilibria that might not be completely mixed. Section \ref{SEC5} contains final remarks on open problems and a discussion of the cases to which the results in this paper do not apply.

\subsection{A motivating example}\label{motivatingexample} We show through the analysis of an example, due to E. Solan and O.N. Solan in \citet{AM2018}, p. 151, how verifying the sign of the index of an equilibrium can be a complicated task with the use of standard  tools. Later, we revisit this example to show how index computation is straightforward applying the new tools. Consider the following three-player game: 
\medskip
\begin{table}[h]
\centering

\begin{minipage}{0.45\linewidth}
\centering

\begin{tabular}{c c c}
\cline{2-3}
 & \multicolumn{1}{|c|}{$a$} & \multicolumn{1}{|c|}{$b$} \\
\cline{1-1}\cline{2-3}
\multicolumn{1}{|c|}{$a$} & \multicolumn{1}{|c|}{$(1,1,1)$}  & \multicolumn{1}{|c|}{$(-5,0,3)$} \\
\cline{1-3}
\multicolumn{1}{|c|}{$b$} & \multicolumn{1}{|c|}{$(0,3,-5)$} & \multicolumn{1}{|c|}{$(0,0,1)$} \\
\cline{1-3}
\end{tabular}

\vspace{0.5em}
\textbf{Player 3 chooses $a$}
\end{minipage}
\hfill
\begin{minipage}{0.45\linewidth}
\centering

\begin{tabular}{c c c}
\cline{2-3}
 & \multicolumn{1}{|c|}{$a$} & \multicolumn{1}{|c|}{$b$} \\
\cline{1-1}\cline{2-3}
\multicolumn{1}{|c|}{$a$} & \multicolumn{1}{|c|}{$(3,-5,0)$} & \multicolumn{1}{|c|}{$(1,0,0)$} \\
\cline{1-3}
\multicolumn{1}{|c|}{$b$} & \multicolumn{1}{|c|}{$(0,1,0)$}  & \multicolumn{1}{|c|}{$(0,0,0)$} \\
\cline{1-3}
\end{tabular}

\vspace{0.5em}
\textbf{Player 3 chooses $b$}
\end{minipage}

\medskip
\end{table}

There is a completely mixed equilibrium $\s^*$ in this game, where all players use equal weights for all strategies. The simplest way to compute the index of this equilibrium is to attempt to use some of the index properties. In this case, we use the property according to which the sum of the indices of all equilibrium components must be $+1$. But this immediately imposes the difficult task of checking whether the game has other equilibrium points. If it does, then, hopefully, other properties of the index of equilibria could eventually be used to pin down the sign of the index of $\s^*$. In the example, it is easy to see that there is a strict equilibrium in $(a,a,a)$, and index theory tells us that such an equilibrium must have index $+1$. Therefore, all other equilibria of the game must have indices whose sum is $0$. If somehow one could check that other equilibria - besides the completely mixed - do not exist, we would conclude that $\s^*$ has index $0$. But as mentioned, even in this game, this is not straightforward. 

The next method one could use is to check whether the Jacobian of a system describing the completely mixed equilibrium of the game $\s^*$ is non-singular. Denote player $n$'s pure strategy by $s_n \in \{a,b\}$, by $\s_n$ a mixed strategy of player $n$, and by $G_n$ player $n$'s payoff function. The following system describes the completely mixed equilibria of the game:

\begin{equation}
\begin{aligned}
G_1(a,\sigma_{-1})-G_1(b,\sigma_{-1})&=0,\\
G_2(a,\sigma_{-2})-G_2(b,\sigma_{-2})&=0,\\
G_3(a,\sigma_{-3})-G_3(b,\sigma_{-3})&=0.
\end{aligned}
\end{equation}

At the completely mixed equilibrium $\sigma^\ast$, every player assigns
probability $1/2$ to each strategy. Translate $\sigma^\ast$ to the
origin by setting
\[
\begin{aligned}
\sigma_1(a)&=\frac12+z, &\qquad \sigma_1(b)&=\frac12-z,\\
\sigma_2(a)&=\frac12+x, &\sigma_2(b)&=\frac12-x,\\
\sigma_3(a)&=\frac12+y, &\sigma_3(b)&=\frac12-y.
\end{aligned}
\]
Thus, $\sigma=\sigma^\ast$ corresponds to $(x,y,z)=(0,0,0)$.
Substitution into the indifference equations gives
\medskip
\begin{align*}
G_1(a,\sigma_{-1})-G_1(b,\sigma_{-1})
&=4(x-y+xy),\\
G_2(a,\sigma_{-2})-G_2(b,\sigma_{-2})
&=4(y-z+yz),\\
G_3(a,\sigma_{-3})-G_3(b,\sigma_{-3})
&=4(z-x+zx).
\end{align*}

Therefore, the system describing the completely mixed equilibria reduces to
\begin{equation}\label{system}
\begin{aligned}
p_1&=x-y+xy=0,\\
p_2&=y-z+yz=0,\\
p_3&=z-x+zx=0.
\end{aligned}
\end{equation}

Verifying that the Jacobian of this system at $0$ is non-singular would imply that the equilibrium has non-zero index. But a quick verification shows that the Jacobian at $0$ is singular. These are the easiest methods to compute the sign of the index of $\s^*$. Another frequently applied method involves payoff perturbations: one could use the property that the index is \textit{locally constant with respect to payoff perturbations}, that is, fixing a neighborhood of $\s^*$ in the mixed strategy set that contains no other equilibrium in its closure, there exists a $\delta>0$ such that for any generic $\delta$-perturbation of the payoffs of the game, there are finitely many equilibria in the interior of the neighborhood, each with an index of $+1$ or $-1$ summing to the index of $\s^*$. The difficulty with using this property is clear: one in principle doesn't know how small $\delta>0$ must be, nor have any idea how many equilibria may arise in the neighborhood; one needs to determine all these equilibria and compute the index of each one separately. 

If one believes $\s^*$ has index $0$, one could attempt to construct a perturbation of the payoffs that has no nearby equilibrium (this is also not straightforward), which is, in general, only a sufficient condition for $0$-index. In the next section we introduce the tools and the formula that will allow us to compute the index of $\s^*$ by analysing simple algebraic properties of the system of polynomials in \eqref{system}.

\section{Preliminaries}\label{SEC2}

We review some basic algebraic and topological facts that are required for this paper. For the algebra, a general practical source for what is not defined in this paper is \cite{CLO2004}. For the topological degree and index theories a standard reference is \cite{D1972}. 

\subsection{Algebraic Preliminaries} A monomial in a collection of $n$-variables $X_1,...., X_n$ is a product $X^{\a} = X^{\a_1}_1...X^{\a_n}_n$, where $\a_i$ is a non-negative integer. Real polynomials in $n$ variables can be written as finite linear combinations of monomials, with real coefficients. Formal real power series comprise the set of infinite linear combinations of monomials with real coefficients. Denote by $\mathbb{R}[[X_1,....,X_n]]$ the set of real formal power series. The set of real formal power series is an example of a commutative ring with unity, where operations of addition and multiplication are the usual operations one generalizes from polynomial rings (which are also commutative rings with unity). 

An ideal $I$ of a ring $R$ is a subset of $R$ that is closed for addition ($x,y \in I \implies x + y \in I$) and absorbs multiplication ($x \in R, y \in I \implies xy \in I$). An ideal $I$ is finitely generated if there exists $f_1,...,f_n \in I$ such that for any element $x \in I$, there exist $a_i \in R, i=1,....,n$, such that $x = \sum_i a_i f_i$. In this case, we denote $I = \langle f_1,...,f_n \rangle$. If $I ,J$ are ideals of $R$, $I+J$ (the sums of elements of each ideal) and $I \cap J$ are ideals of $R$. For a finite positive integer $k$, $I^k$ denotes the smallest ideal generated by the product of $k$ elements of $I$. An ideal is called \textit{square-zero} if $I^2 = 0$. An ideal that is generated by a single element is called \textit{principal}. 

Given an ideal $I \subset R$,  say that we define an equivalence relation $\sim$ in $R$,  by letting $x \sim y \iff x-y \in I$. The equivalence class of $x$ is denoted $[x]$. The set of equivalence classes under this relation is denoted $R/I$ or $\frac{R}{I}$, and is given a ring structure from the operations of $R$: $x,y \in R, [x] \oplus [y] := [x+y]$ and $z,x \in R, [z] \odot [x] := [zx]$. If $I \subset J$ are ideals of $R$, then $J/I$ is also defined as the set of equivalence classes of elements of $J$ and is an ideal of $R/I$.

In this paper we will mostly consider ideals of real formal power series and quotient rings of such ideals. All ideals of real formal power series are finitely generated, as well as the quotient of such ideals. 

A function $\varphi: R \to L$ between rings $R$ and $L$ is a ring-homomorphism if $\forall x,y,a,b \in R, \varphi(ax + by) = \varphi(a)\varphi(x) + \varphi(b)\varphi(y)$ and $\varphi(1_R) = 1_{L}$. The \text{kernel of $\varphi$}, denoted Ker$(\varphi)$ is a subset of $R$ composed of the elements of $R$ that are mapped under $\varphi$ to $0$. Note that Ker$(\varphi)$ is an ideal of $R$. An \textit{isomorphism of rings} is a ring-homomorphism that has an inverse ring-homomorphism. Intuitively, two isomorphic rings are identical as rings, up to relabeling of their elements. We denote the existence of a ring-isomorphism between rings $R$ and $L$ by $R \cong L$. 

Real polynomials have a natural real vector space structure: they are trivially closed by addition (we have already observed that from their ring structure), but also by scalar multiplication by an element of $\mathbb{R}$, thus satisfying the properties of an $\mathbb{R}$-algebra (i.e., a commutative ring with unity that contains $\mathbb{R}$).  By the same reason, the real formal power series are an $\mathbb{R}$-algebra as well. Therefore, they are both rings and real vector spaces.  The dimension of a real vector space $V$ will be denoted $\text{dim}_{\mathbb{R}}V$. 
We say that two $\mathbb{R}$-algebras $R$ and $L$ are homomorphic if there exists an $\mathbb{R}$-linear homomorphism $\psi: R \to L$ such that $\psi(xy)=\psi(x)\psi(y)$. The homomorphism is unital if it takes the multiplication unity element of $R$ to the multiplication unity of $L$. All algebra-homomorphisms that appear in this paper are unital. An isomorphism of $\mathbb{R}$-algebras (i.e., an $\Re$-algebra homomorphism with an inverse $\Re$-algebra homomorphism) is  denoted by $\cong_{\Re}$. Note that an ideal $I$ of an $\mathbb{R}$-algebra that is closed by scalar multiplication is also a linear subspace of this algebra.

The ring of real formal power series is an example of a \textit{local ring}: a ring that has exactly one maximal ideal (a maximal ideal is a proper ideal of a ring that is not properly contained in any other ideal). The maximal ideal of  $\mathbb{R}[[X_1,....,X_n]]$ is the ideal $\mathfrak{m} = \langle X_1,...,X_n \rangle$. Note that if $f \notin \mathfrak{m}$, then $f$ can be inverted - i.e., one can explicitly write down a real formal power series $g$ which is the multiplication inverse of $f$. This simple fact is crucial for solving Example \ref{example}.

If a list of $m$ real formal power series in $n$ variables $p_1, ....,  p_m$ is such that for each $i$, $p_i$ has no constant term, then it follows that $\langle p_1,..., p_m \rangle \subset \mathfrak{m}$. In particular, if for each $i$, $p_i =0$ is a system of polynomials with $0$ as an isolated root, then each $p_i$ has no constant term.

For an ideal $I \subset \mathfrak{m} \subset \Re[[X_1,....,X_n]]$, $\Re[[X_1,...,X_n]]/I$ is a ring, but also inherits a real vector space structure from $\Re[[X_1,...,X_n]]$, which makes it into an $\Re$-algebra as well. Determining which class $[f]$ is non-zero in this quotient is equivalent to determining if $f \in I$. This is relevant, for example, if one wants to know if a linear combination in $\Re[[X_1,...,X_n]]/I$ is zero, or if one wants to construct a basis for the vector space $\Re[[X_1,...,X_n]]/I$ in order to establish its dimension, something that is intimately linked to the problem of computing the index of an isolated equilibrium (see Proposition \ref{auxprop}). For our goal in this paper, we will be interested in the case where $I$ is finitely generated by polynomials $p_1,....,p_m$. Therefore we need to establish that there exists $a_i \in \Re[[X_1,...,X_n]], i=1, ..., m$ such that $f = \sum_{i}a_i p_i$. Mora's normal form algorithm allows us to answer this question, by generalizing the known division algorithms for polynomials (see \cite{CLO2004}, ch.4, Corollary 3.14 and exercise 2 in $\textsection$ 4, which shows one can perform division with real formal power series as inputs and a local order). In order to do that, we need to define a local order on the set of monomials: an order $>$ on the set of monomials is a \textit{local order} if it is total, compatible with multiplication ($\a, \b, \g \in \mathbb{Z}^n_+, X^{\a} > X^{\b} \implies X^{\a + \g} > X^{\b + \g}$) and for each $i=1,...,n, 1 > X_i$. A particular local order we use is called the \textit{degree-anticompatible lex order}: in this order $\a, \b \in \mathbb{Z}^n_+, X^{\a} >_{a\ell ex} X^{\b}$ if $$ |\a| = \sum_i \a_i < |\b| = \sum_i \b_i$$ or if 

$$ |\a| = |\b| \text{ and } X^{\a} >_{\ell ex} X^{\b}, $$ where  $>_{\ell ex}$ is some lexicographic order on monomials.

With a local order, one can define the leading term of a real formal power series: if $f = \sum_{\a \in \mathbb{Z}^n_+}c_{\a}X^{\a}$, then its leading term LT$(f)$ is the $c_{\bar \a}X^{\bar \a}$ where $X^{\bar \a}$ is maximal under the local order $>$ over the monomials defining $f$. Given an ideal $I \subset \Re[[X_1,...,X_n]]$, the \textit{leading term ideal of I} -- denoted $\langle LT[I] \rangle$ -- is the ideal generated by the leading terms of the power series in $I$. This is an example of a monomial ideal, an ideal that has a set of generators composed of monomials. The leading term ideal has a few crucial properties that allow us to quickly perform the verification of whether an index is $0$ or not: the first is given by Theorem 4.3 in \cite{CLO2004}, namely, when the dimension of $\Re[[X_1,...,X_n]]/I$ is finite, the dimension of $\Re[[X_1,...,X_n]]/I$ is equal to the dimension of $\Re[[X_1,...,X_n]]/ \langle LT[I] \rangle$ and, moreover, those monomials that do not belong to $\langle LT[I] \rangle$ form a basis for the real vector space $\Re[[X_1,...,X_n]]/I$. Verifying that a polynomial is a member of a monomial ideal is usually a simpler task than doing it when the ideal does not have such a property, and we will explore precisely this in Example \ref{example}. In addition, the dimension of $\Re[[X_1,...,X_n]]/ \langle LT[I] \rangle$ is simply the number of standard monomials, i.e., those monomials $X^{\a} \notin \langle LT[I] \rangle$ (that is, monomials that are not divisible by any of the generators in $LT[I]$).

Note that when we consider real formal power series rings in one variable, there is only one local order ($1 > X > X^2 > ....$). The leading term of such a power series $f$ is then the term that has the smallest degree, also known as the order of $f$, denoted \text{ord($f$)}.

\subsection{Topological Preliminaries for Finite Games}

We start by recalling the definition of the Brouwer degree of a compact and connected set of zeroes of a continuous map $f: (W, Q) \to (\Re^m, 0)$, where $W \subset \Re^m$ is an open set containing the compact and connected subset $Q = f^{-1}(0)$.\footnote{The assumption of connectedness is only to simplify exposition, since we are interested in topological degrees of connected components of equilibria only in this paper.} We follow \cite{D1972}, ch. IV, sec. 5. Other standard references are \cite{RB2014} and \cite{AM2018}.

\begin{definition}
Fix $f^{-1}(0) = Q \subset W \subset \mathbb{R}^{m}$ where $Q$ is compact and connected and $W$ is open. Recall that the one-point compactification of $\mathbb{R}^{m}$ is (homeomorphic to) an $m$-dimensional sphere. The \textit{fundamental class} $\mathcal{O}_{Q} \in H_m(W,W \setminus Q)$ is the image of $1 \in \mathbb{Z}$ under the composition ($H_*$ denotes the singular homology functor)
$$\mathbb{Z} \cong H_m(\mathbb{S}^m) \rightarrow H_m(\mathbb{S}^{m}, \mathbb{S}^{m}  \setminus Q) \rightarrow H_m(W,W \setminus Q)$$

where the first and second arrows are the homomorphisms induced by inclusion and the second one is an isomorphism, by excision.   
\end{definition} 

Recall that $H_m(\mathbb{R}^{m}, \mathbb{R}^{m}  \setminus \{0\}) \cong \mathbb{Z}$, so we call its generator $\mathcal{O}_0$. We define the topological degree of $Q$ w.r.t. $f$ at $0$ as the (unique) integer number $\text{deg}_0(Q,f)$ defined by $f_*(\mathcal{O}_{Q}) = \text{deg}_0(Q,f) \mathcal{O}_0$, where $f_*$ is the homology-induced homomorphism $f_* : H_m(W,W \setminus Q) \to H_m(\Re^{m}, \Re^{m}  \setminus \{0\})$. 

If at $x \in Q$ the map $f$ is differentiable with determinant Jacobian distinct from zero, then $\{x\} = Q$ (by the Inverse Function Theorem) and the topological degree of $Q$ equals the sign of the determinant Jacobian at $x$, i.e., $-1$ or $+1$. When the determinant Jacobian is zero, however, the topological degree of a singleton $Q$ might not be zero. We are interested in this paper in the topological degree of isolated zeros of polynomial maps that describe the completely mixed equilibria of finite games (Proposition \ref{auxprop} provides a characterization of the degree based on algebraic properties of the polynomial map). 

The topological degree is the main tool used to define the concept of index of fixed points, which we now present. Let $g: U \to \Re^m$ be a continuous map where $U \subset \Re^m$ is open.  Let $K \subset U$ be a compact and connected set of fixed points of $g$. Letting $d := i - g$, where $i$ is the inclusion of $U$ in $\Re^m$, the index of $K$ w.r.t. $g$, denoted $\text{ind}(K,g)$ is the topological degree $\text{deg}_0(K,d)$ defined by 

$$ d_* : H_m(U,U \setminus K) \to H_m(\Re^{m}, \Re^{m}  \setminus \{0\}) $$ $$ \mathcal{O}_K \mapsto \text{deg}_0(K,d) \mathcal{O}_0.$$

We are now ready to define the index of equilibria. For this purpose, let $\mathcal{N} = \{1,...,N\}$ denote the set of players, $S_n = \{s^0_{n},...,s^{k_n}_{n}\}$ denote the set of pure strategies of player $n$, with $S = \times_{n \in \N}S_n$,  and $G_n: S \to \mathbb{R}$ denotes a payoff function of player $n$. The mixed strategy set of player $n$ is denoted $\S_n$, with $\S = \times_n \S_n$. Let $\Gamma := \Re^{|S_1|...|S_{N}|N}$ denote the set of possible payoffs for the game defined by players $\N$, and pure strategy set $S$ for the players.

\begin{definition}
A \textit{Nash map} is a continuous function $f: \G \times \S \rightarrow \S$ such that for each $G \in \G$, the fixed points of its restriction $f_{G}$ from $\{G\}\times \S$ to $\S$, viewed as a map from $\S$ to itself, are the Nash equilibria of $G$.
\end{definition}

One example of a Nash map is the map used in \cite{JN1951} to prove existence of equilibria. Another is the $GPS$-map in \cite{GPS1993}.

Define $J := \times_n (\S_n) $, where $(\S_n)$ is the affine space generated by $\S_n$, and let $J_0 := J - J$, where the symbol ``$ - $'' denotes the subtraction in the vector space $\times_n \mathbb{R}^{S_n}$. The space $J_0$ is the unique linear subspace of $\times_n \mathbb{R}^{S_n}$ that is parallel to $J$ and of the same dimension as $J$.  Both $J$ and $J_0$ are homeomorphic to the same Euclidean space. Therefore the definition of fundamental class applies in the obvious way to compact subsets of these spaces. Let $r$ be a retraction of $J$ onto $\S$. Then every Nash map $f$ extends to the map $f \circ (\text{id}_{\G}\times r)$ on $\Gamma \times J$. Let $E(G)$ denote the set of equilibria of $G \in \G$. Recall that equilibria of finite games have finitely many connected components. 


\begin{definition}\label{def index}
Fix a Nash map $f$, a finite game $G$, and a connected component $Q$ of $E(G)$. Choose an open neighborhood $W$ of $Q$ in $J$ disjoint from $E(G) \setminus Q$. The \textit{index of $Q$} - denoted ind$(Q)$ - is defined as \text{ind}$(Q, f_G)$.  \end{definition}

\begin{remark} Any Nash map assigns the same index to $Q$ (cf. \cite{DG2000}). If $r'$ was another retraction of $J$ onto $\S$ then, using a linear homotopy, $r$ is homotopic to $r'$. This implies that the induced homotopy between $f \circ (\text{id}_{\G} \times r)$ and $f \circ (\text{id}_{\G} \times r')$ preserves the set of fixed points. This in particular implies that the indices defined according to both maps agree (cf. \cite{RB2014}, Theorem 9.5). The relevance of the index of fixed points in game theory comes from their implied refinement properties: one immediate consequence of the index being non-zero is that the component is \textit{payoff-robust} (that is, for any neighborhood $U \subset \S$ of $Q$, there exists a neighborhood $V \subset \Re^{S}$ of $G$ such that for any $G' \in V$ there exists an equilibrium $\s' \in U$ of $G'$). This is due to \cite{KR1994} -- see section 2 in \cite{LP2023} for a discussion. \end{remark}

The next proposition provides a formula for the computation of the absolute value of the topological degree of an isolated root of a polynomial system, which is directly based on the formula obtained by Eisenbud and Levine in Theorem 1.1 in \cite{EL1977}. In their formula, the $\mathbb{R}$-algebras considered are germs at $0$ of smooth real maps and the condition for the formula to hold is that the germ is finite. To make the application of the formula straightforward, it is worth recasting the formula in terms of $\mathbb{R}$-algebras of formal power series, and the condition in terms of finite dimensionality of the $\mathbb{R}$-algebras. The proposition is certainly known, but we include a proof here for completeness. In the proof of Theorem \ref{finalthm}, we employ another formula for the computation of the degree presented in Theorem 1.2 in \cite{EL1977}.

\begin{proposition}\label{auxprop}
Suppose $f: (\R^n,0) \to (\R^n, 0)$ is a polynomial map with an isolated root at $0$. Suppose $I = \langle f_1,...,f_n \rangle$ is an ideal of  $\mathbb{R}[[X_1,....,X_n]]$ and $\mathbb{R}[[X_1,....,X_n]] / I$ has finite dimension. Then $$\vert \text{deg}_0(0,f) \vert = \text{dim}_{\mathbb{R}}(\mathbb{R}[[X_1,....,X_n]] / I) - 2 \text{dim}_{\mathbb{R}} I_{max},$$ where $I_{max}$ is a maximal square-zero ideal of $\mathbb{R}[[X_1,....,X_n]] / I$.
\end{proposition} 

\begin{proof} Let $S = C^{\infty}_0(\mathbb{R}^n)$ (the $\mathbb{R}$-algebra of germs of smooth functions at $0$) and $A = \mathbb{R}[[X_1,...,X_n]]$. We denote by $I_S$ the ideal generated in $S$ by the $f_i$'s and analogously by $I_A$ the ideal generated in $A$ by the same $f_i$'s.  Let $T: S \to A$ be the ring homomorphism associating to each smooth function germ $f$ at $0$ its Taylor series at $0$. This map is surjective, by Borel's theorem. Hence $T$ induces a surjective map $\hat{T}: S \to \frac{A}{I_A}$ by composing $T$ with the quotient map from $A$ to $\frac{A}{I_A}$. Now, since each generator $f_i$ is polynomial, $T(f_i) = f_i$, which implies that $T(I_S) \subset I_A$. Therefore, the induced map $\bar{T}: \frac{S}{I_S} \to \frac{A}{I_A}$ is still surjective. In order to prove that $\bar{T}$ is actually an isomorphism, we need to prove that Ker($\bar T$) $= \{0\}$. To do that we establish the following lemma.

\begin{lemma}\label{auxproof1} $Ker(\bar{T}) = \frac{Ker(T) + I_S}{I_S}$.\end{lemma}

\begin{proof}[Proof of the Lemma]We first show that Ker$(\bar T) \subset \frac{Ker(T) + I_S}{I_S}$: let $[g] \in \text{Ker}(\bar T)$. Then $[Tg] = [0] \iff Tg \in I_A$. Therefore, $Tg = \sum^{n}_{i=1}a_if_i, (a_i \in A)$. By surjectivity of $T$, we can choose smooth germs $b_i \in S$ with $Tb_i = a_i$. Therefore, $T(g - \sum^{n}_{i=1}b_if_i) = Tg - \sum^{n}_{i=1}a_if_i =0$. Therefore, $g - \sum^{n}_{i=1}b_if_i \in \text{Ker}(T)$. So $g \in Ker(T) + I_S$, which implies that $[g] \in \frac{Ker(T) + I_S}{I_S}$. This shows that $Ker(\bar T) \subset  \frac{Ker(T) + I_S}{I_S}$. The converse is straightforward. \end{proof} 

Given Lemma \ref{auxproof1}, $\text{Ker}(\bar T) = \{0\}$ if and only if $\text{Ker}(T) \subset I_S$. 

We now establish that $\text{Ker}(T) \subset I_S$. First recall that Ker$(T)$ comprises the flat germs, that is, germs of smooth functions at $0$ such that the derivatives at all orders vanish at $0$. Also recall that $S$ is a local ring with maximal ideal denoted by $\mathfrak{m}_S$ generated by $X_1, \ldots, X_n$. Using the known characterization of Ker$(T)$, $$ Ker(T) = \bigcap_{k \geq 1} \mathfrak{m}^k_S,$$ it is sufficient to prove that there exists $k \in \mathbb{N}$ for which $\mathfrak{m}^k_S \subset I_S$. Let $\mathfrak{m}$ denote the ideal $\langle X_1, \ldots, X_n \rangle$ in $\mathbb{R}[X_1, \ldots, X_n]$ and $R$ denote the localized ring of polynomials $\mathbb{R}[X_1,\ldots, X_n]_{\mathfrak{m}}$. Let $I_R$ be the ideal generated by $f_1,\ldots,f_n$ in  $\mathbb{R}[X_1,\ldots,X_n]_{\mathfrak{m}}$. By assumption, as $\text{dim}_{\mathbb{R}}\frac{A}{I_A}$ is finite, Buchberger's algorithm and Theorem 4.3 in \cite{CLO2004}, imply that $\frac{A}{I_A}$ and $\frac{R}{I_R}$ have the same dimension as real vector spaces and have the same standard monomials as well. This implies that there exists $k \in \mathbb{N}$ such that $\mathfrak{m}_{R}^{k} \subset I_R$, where $\mathfrak{m}_R$ is the ideal of $R$ generated by $X_1, \ldots , X_n$. Let now $X^{\alpha}$ be a monomial such that $|\alpha| =k$. Recall that such monomials generate $\mathfrak{m}^k_S$ and $\mathfrak{m}^k_R$. Since $X^{\alpha} \in I_{R}$, then $X^{\alpha} = \frac{f}{g}$, $f \in J, g \in \Re[X_1, \ldots, X_n] \setminus \mathfrak{m}$, and $J \subset \mathbb{R}[X_1,\ldots ,X_n]$ is the ideal generated by $f_i, i=1, \ldots, n$. Therefore, $\frac{f}{g}$ can be written as $\sum^n_{i=1}\frac{a_i}{g}f_i, a_i \in \mathbb{R}[X_1,\ldots ,X_n]$. This implies $\frac{f}{g} \in I_S$,  since $\frac{a_i}{g}$ represents a smooth function germ at $0$. Hence,  $\mathfrak{m}^k_S \subset I_S$ thus concluding the proof that $\bar{T}$ is a ring-isomorphism. Note that, in addition, $\bar T$ is $\Re$-linear, so $\bar{T}$ is an isomorphism of $\Re$-algebras. This implies that both $\frac{S}{I_S}$ and $\frac{A}{I_A}$ have the same dimension. In particular, the maximal square-zero ideal $J_{max}$ of $\frac{S}{I_S}$ maps under $\bar{T}$ to the maximal square-zero ideal $I_{max}$ of $\frac{A}{I_A}$, with both having the same dimension as well. Applying then Theorem 1.1 in \cite{EL1977}, the result follows.\end{proof} 


\begin{remark} If $0$ is isolated as a complex root of $f$, then one can show that $\mathbb{R}[[X_1,....,X_n]] /I$ has finite dimension - thus making the assumption of finite real-dimension of  $\mathbb{R}[[X_1,....,X_n]] / I$ redundant. However, if $0$ is isolated as a real root, but not as a complex root, the assumption of finite dimension of $\mathbb{R}[[X_1,....,X_n]] /I$ is still required to obtain Proposition \ref{auxprop}: take as an example $f(x,y) = ((x^2+y^2)x, (x^2+y^2)y)$. The point $0$ is an isolated real root (but not as a complex root) and $\mathbb{R}[[X_1,....,X_n]] /I$ has infinite dimension. \end{remark} 


\section{Index and Robustness}\label{SEC3}

We start by establishing the required notation.  
Let $G = (\mathcal{N}, (S_n)_{n \in \mathcal{N}}, (G_n)_{n \in \mathcal{N}})$ denote a finite game, where $\mathcal{N} = \{1,...,N\}$ denotes the set of players, $S_n = \{s^0_{n},...,s^{k_n}_{n}\}$ denotes the set of pure strategies of player $n$, with $S = \times_{n \in \N}S_n$,  and $G_n: S \to \mathbb{R}$ denotes the payoff function of player $n$. To avoid notational clutter, the mixed extension of the payoff function of player $n$ is denoted by $G_n$ as well, and the mixed strategy set of player $n$ is denoted by $\S_n$, with $\S = \times_{n \in \mathcal{N}}\S_{n}$. The interior of $\S$ relative to the affine space generated by $\S$ is denoted by int$(\S)$. For each $\s$, let $G^n(\s) = (G_n(s_n, \s_{-n}))_{s_n \in S_n}$. Let $\s^*$ be an isolated completely mixed equilibrium of $G$. For each $s_n \in S_n$, let $e_{s_n}$ be the canonical vector of $\Re^{S_n}$ with $1$ in the $s_n$-th position.  Let $q: \R^{\kappa+N} \to \R^{\kappa + N}$ be a polynomial map defined by its coordinate maps as follows: $q_{s^{i}_n}(\s) = G_n(e_{s^0_n}, \s_{-n}) - G_n(e_{s^i_n}, \s_{-n})$, $i=1,...,k_n$, and $q_{0,n}(\s) = 1 - \sum^{k_n}_{i=0}\s_{s^i_n}$, $\forall n \in \N$, with $\kappa = \sum_{n \in \mathcal{N}} k_n$. Using the fact that $\s_n$ is a probability distribution, we can eliminate the variable $\s_{s^0_n}$ by setting $\s_{s^0_n} = 1 - \sum^{k_n}_{i=1} \s_{s^i_n}$. Now, consider the change of variables $\s_{s^i_n}\mapsto X_{s^i_n} + \s^*_{s^i_n}, i=1,...,k_n$ and let $p: \mathbb{R}^\kappa \to \mathbb{R}^\kappa$ be the associated map with an isolated root at zero. The construction of $p$ generalizes the construction of the system \eqref{system} in Example \ref{motivatingexample}. This is a multiaffine polynomial with an isolated root at $0$.

We need a known specific result on the index of equilibria for the proof of part (b) of Theorem \ref{mainthm1} that we will present in this section. Let $r_n: \mathbb{R}^{S_n} \to \S_n$ denote the closest-point retraction from a point $z_n \in \Re^{S_n}$ to a point in $\S_n$, and let $r = \times_{n}r_n$. Letting $\mathcal{E}_G = \{ (g,\s) \in \Re^{\kappa + N} \times \S \mid \s \text{ is an equilibrium of } G \oplus g  \}$, where $G \oplus g$ is a finite game, with the mixed strategy set of each player $n$ as $\S_n$ and payoffs defined as  $G_n(\s) + \s_n \cdot g_n$, $\s_n \in \S_n$, $g_n \in \Re^{k_n + 1}$, and $\cdot$ the Euclidean inner product. The set $\mathcal{E}_G$ defines a restricted Nash equilibrium graph, where one leaves the payoffs of $G$ fixed and varies only the payoff perturbations $g \in  \Re^{\kappa + N}$. 

Define $\theta_G: \mathcal{E}_G \to \Re^{\kappa + N}$ by $\theta_G(g, \s) = (\s_n + G^n(\s) + g_n)_{n \in \N}$. Using a similar argument as in Kohlberg and Mertens, one can obtain that $\theta_G$ is a homeomorphism, with inverse $\theta_G^{-1}(z) = (g(z), r(z))$, where $g_n(z) = z_n - r_n(z_n) - G^{n}(r_{-n}(z_{-n}))$, and $g = \times_n g_n$. Letting proj: $\mathcal{E}_G \to \Re^{\kappa + N}$ denote the projection over the payoffs $g$, for each isolated equilibrium $\s^*$ of $G$, and $U \subset \Re^{\kappa + N}$ a neighborhood of $z^* = \theta_G(0,\s^*)$ containing  in its closure no other $z$ such that $\theta_G(0,\s) = z$, the topological degree over $0$ of $\text{proj} \circ \theta^{-1}_G: U \to \Re^{\kappa + N}$ equals the index of $\s^*$ (see \textit{Formulation} and \textit{Appendix A} in \cite{GW2005} for a detailed explanation of this fact).

The next proposition tells us that we can use the system of payoff indifferences given by the polynomial map $p$ we defined above to compute the index of the completely mixed equilibrium $\s^*$ associated to it. 

\begin{proposition}\label{degindexequal} Given any finite game $G$ with $\s^*$ a completely mixed and isolated equilibrium of $G$, it follows that the topological degree of $0$ with respect to $p$ equals the index of $\s^*$. \end{proposition} 

\begin{proof}Define $$f_{s_n}(\s) = \frac{\s_{s_n} + G^{n}(\s) \cdot (e_{s_n} - e_{s^0_n})}{1 + \sum_{\tilde s_n \in S_n}G^n(\s)\cdot(e_{\tilde s_n} - e_{s^0_n})}.$$ Let $f_n = \times_{s_n \in S_n} f_{s_n}$, and $f = \times_n f_n$. Let $U$ be a closed neighborhood of $\s^*$ in int$(\S)$ such that $f(U) \subset \S$, and such that it contains no equilibria of $G$ besides $\s^*$. Then $f: U \to \S$ is well-defined and is such that $\s \in U$ is a fixed point of $f$ if and only if $\s = \s^*$. 

Let $$h_{s_n}(\s) = \frac{\s_{s_n} + \text{max} \{0, G^{n}(\s) \cdot (e_{s_n} - \s_n)\} }{1 + \sum_{\tilde{s}_n \in S_n}\text{max} \{0, G^{n}(\s) \cdot ( e_{\tilde{s}_n} - \s_n)\}},$$ 

with $h_n = \times h_{s_n}: \S \to \S$, and $h = \times_n h_n$, be the map used in \cite{JN1951} to prove existence of Nash equilibria. In what follows, a homotopy $H: [0,1] \times U \to \S$ between $H(0, \cdot ) = f: U \to \S$ and $H(1, \cdot) = h: U \to \S$ is called admissible if the boundary of $U$ contains no fixed points of the homotopy, i.e., for each $\s \in \partial U$ and $t \in [0,1]$, $H(t,\s) \neq \s$. Admissible homotopies preserve the fixed points index (cf. \cite{D1972}, VII, 5.8).

\begin{lemma}\label{smalemma}There is an admissible homotopy between $f: U \to \S$ and $h: U \to \S$. Therefore, $\text{ind}(\s^*, f) = \text{ind}(\s^*, h)$. \end{lemma}

\begin{proof}[Proof of the Lemma] Let $\phi_{s_n}(\s) = G^{n}(\s)\cdot(e_{s_n} - e_{s^0_n})$ and $q_{s_n}(\s) = \text{max}\{0, G^n(\s)\cdot (e_{s_n} - \s_n)\}$. Adjusting the neighborhood $U$ if necessary to a smaller one around $\s^*$, define the homotopy $H: [0,1] \times U \to \S$,  $$H_{s_n}(t,\s) = \frac{\s_{s_n} + (1-t)\phi_{s_n}(\s) + tq_{s_n}(\s)}{1 + \sum_{s_n \in S_n}(1-t)\phi_{s_n}(\s) + tq_{s_n}(\s)}.$$ Let $0<t<1$. Suppose by contradiction $\s$ is a fixed point of the homotopy distinct from $\s^*$. Therefore, $\s$ is not a Nash-equilibrium of $G$. Fix $n \in \mathcal{N}$ such that $\s_n$ is not a best-reply to $\s$. Suppose first $s^0_n$ is a best-reply to $\s$ in $G$. Then $\phi_{s^0_n}(\s) =0$ and $q_{s^0_n}(\s)>0$. Therefore, $\sum_{s_n}(1-t)\phi_{s_n}(\s) + tq_{s_n}(\s)>0.$ There exists now $\bar s_n \in S_n$ a pure strategy that yields an inferior payoff against $\s_{-n}$ than $\s_n$ (since $\s_n$ is completely mixed). Therefore, $\phi_{\bar s_n}(\s) < 0$ and $q_{\bar s_n}(\s) =0$, thus implying $\sum_{s_n}(1-t)\phi_{ s_n}(\s) + tq_{ s_n}(\s)<0,$ which is a contradiction. Therefore, $s^0_n$ cannot be a best-reply to $\s$. Let now $\hat s_n$ be a best-reply to $\sigma$, which implies that $\sum_{s_n}(1-t)\phi_{s_n}(\s) + tq_{s_n}(\s)>0.$ Let now $\tilde s_n$ yield a payoff against $\s$ that is at most as good as $s^0_n$ and inferior to $\s_n$. This implies then $\sum_{s_n}(1-t)\phi_{s_n}(\s) + tq_{s_n}(\s)\leq 0$. Contradiction. Therefore, the only fixed point of $H(t, \cdot)$ is $\s^*$ and the homotopy is in particular admissible.\end{proof} 


For each $n \in \mathcal{N}$, let $A^0_n = \{ X_n = (X_{s_n})_{s_n \in S_n} \in \Re^{S_n} \mid \sum_{s_n \in S_n}X_{s_n} =0 \}$, and $A^0 = \times_n A^0_n$. Let $F_n = \{ X \in \Re^{S_n} \mid \sum_{s_n \in S_n}X_{s_n} -1 =0\}$. Consider the affine homeomorphism $T_n: A^0_n \to F_n$, where $T_n = \times_{s_n}T_{s_n}$, $T_{s_n}(X_n) = X_{s_n} + \s^{*}_{s_n}$. Let $V = T^{-1}(U)$. Let $T = \times_n T_n$. Define $\hat f = T^{-1} \circ f \circ T: V \to A^0$. By the commutativity property of the index (cf. \cite{D1972}, VII, 5.9), $\text{ind}(0, \hat f) = \text{ind}(\s^*, f)$. Note now that for each $s_n \in S_n$,

$$\hat f_{s_n} (X) =  \frac{T_{s_n}(X) + G^{n}(T(X)) \cdot (e_{s_n} - e_{s^0_n})}{1 +\sum_{\tilde s_n \in S_n}[G^{n}(T(X)) \cdot (e_{\tilde s_n} - e_{s^0_n})]} - \s^*_{s_n}.$$

Consider now the map $S: \Re^{\kappa} \to A^0$ given by $S = \times_n S_n$ with $$S_n(X_{s^1_n}, \ldots, X_{s^{k_n}_n}) = (- \sum^{k_n}_{i=1} X_{s^i_n}, X_{s^1_n}, \ldots, X_{s^{k_n}_n}).$$ The map is a homeomorphism with inverse given by $S^{-1}$. Define $\hat h = S^{-1} \circ \hat f \circ S$. Again, by the commutativity property of the index, it follows that $\text{ind}(0, \hat h) = \text{ind}(0, \hat f)$. Let now $W = S^{-1}(V)$. Consider the homotopy $Q: [0,1] \times W \to \mathbb{R}^{\kappa}$ given by: $$Q_{s_n}(t, X) =   \frac{T_{s_n}(S(X)) + G^{n}(T(S(X))) \cdot (e_{s_n} - e_{s^0_n})}{1 + (1-t)\sum_{s_n \in S_n}[G^{n}(T(S(X))) \cdot (e_{s_n} - e_{s^0_n})]} - \s^*_{s_n}.$$

A similar argument to the one in the proof of Lemma \ref{smalemma} yields that $0$ is the only fixed point of the homotopy which then implies that ind$(0, \hat h) =$ ind$(0, Q(1,\cdot)$). Now, notice that $p_{s_n}(X) = X_{s_n} - Q_{s_n}(1,\cdot)$, so the displacement of $Q(1, \cdot)$ equals $p$. Since the index of $0$ w.r.t. $Q(1, \cdot)$ equals the topological degree of $0$ w.r.t. its displacement (by definition), then the result follows. \end{proof}

\begin{example}\label{example2} We revisit the example in subsection \ref{motivatingexample} in order to illustrate the application of the formula of Proposition \ref{auxprop} to the computation of the index of $\s^*$. 
\medskip
\medskip

\begin{table}[h]
\centering

\begin{minipage}{0.45\linewidth}
\centering

\begin{tabular}{c c c}
\cline{2-3}
 & \multicolumn{1}{|c|}{$a$} & \multicolumn{1}{|c|}{$b$} \\
\cline{1-1}\cline{2-3}
\multicolumn{1}{|c|}{$a$} & \multicolumn{1}{|c|}{$(1,1,1)$}  & \multicolumn{1}{|c|}{$(-5,0,3)$} \\
\cline{1-3}
\multicolumn{1}{|c|}{$b$} & \multicolumn{1}{|c|}{$(0,3,-5)$} & \multicolumn{1}{|c|}{$(0,0,1)$} \\
\cline{1-3}
\end{tabular}

\vspace{0.5em}
\textbf{Player 3 chooses $a$}
\end{minipage}
\hfill
\begin{minipage}{0.45\linewidth}
\centering

\begin{tabular}{c c c}
\cline{2-3}
 & \multicolumn{1}{|c|}{$a$} & \multicolumn{1}{|c|}{$b$} \\
\cline{1-1}\cline{2-3}
\multicolumn{1}{|c|}{$a$} & \multicolumn{1}{|c|}{$(3,-5,0)$} & \multicolumn{1}{|c|}{$(1,0,0)$} \\
\cline{1-3}
\multicolumn{1}{|c|}{$b$} & \multicolumn{1}{|c|}{$(0,1,0)$}  & \multicolumn{1}{|c|}{$(0,0,0)$} \\
\cline{1-3}
\end{tabular}

\vspace{0.5em}
\textbf{Player 3 chooses $b$}
\end{minipage}

\medskip
\end{table}

\medskip
\medskip

This game has a completely mixed equilibrium $\s^*$ where each player randomizes with equal probabilities between $a$ and $b$. From this we can compute the polynomial map $p$ corresponding to this completely mixed equilibrium: as in subsection \ref{motivatingexample}, let $x=\s_2(a)-\frac{1}{2}$, $y=\s_3(a)-\frac{1}{2}$, and $z=\s_1(a)-\frac{1}{2}$. We then have the following system of equations: 

\begin{equation} 
\begin{aligned} 
 p_1 &= x-y+xy = 0 \\ p_2 &= y-z+yz = 0 \\ p_3 &= z-x+zx =0
\end{aligned} 
\end{equation} 
\medskip

 In order to verify that $\s^*$ has index $0$, we use the fact that the algebra given by $\mathbb{R}[[x,y,z]]/\langle p_1, p_2, p_3 \rangle$ has the same dimension as $Q = \mathbb{R}[[x,y,z]]/ \langle LT[ \langle p_1, p_2, p_3 \rangle]\rangle $, and the latter is immediate to compute: fix the degree anticompatible lex order given by $x > y > z$: then, immediately from $p_1$ and $p_2$, $x \in LT[\langle p_1, p_2, p_3 \rangle]$, $y \in LT[\langle p_1, p_2, p_3 \rangle]$. Therefore, a basis for $Q$ is given by $\{1,z,...,z^{k}\}$, for some non-negative integer $k$. We now can verify that $k=1$: from $p_2$, $y(1+z) -z = 0$, yielding $y = \frac{z}{1+z}$. From $p_3$, $x(z-1) + z =0$, yielding $x = \frac{z}{1-z}$. Now substituting in $p_1$, it yields $\frac{3z^2}{(1-z)(1+z)}$, and since $\frac{3}{(1-z)(1+z)}$ is a unit in $\mathbb{R}[[x,y,z]]$, it follows that $\{1,z\}$ is a basis for $Q$, since $z^2 \in LT[ \langle p_1, p_2, p_3 \rangle]$. Now, note that the ideal $\langle z \rangle$ is evidently the maximal square-zero ideal of $\mathbb{R}[[x,y,z]]/\langle p_1, p_2, p_3 \rangle$. Applying the formula of Proposition \ref{auxprop} to $p = (p_1, p_2, p_3)$ and Proposition \ref{degindexequal} give that the index of the completely mixed equilibrium $\s^*$ is zero. \end{example}

As observed in the introduction, the non-zero index of an equilibrium component is a sufficient condition only for payoff-robustness. Using Proposition \ref{auxprop} and Proposition \ref{degindexequal} we can define a class of games for which non-zero index and payoff-robustness are equivalent conditions. Moreover, this is a class of games for which the indices that can occur in completely mixed equilibria parallel those that occur in general one-dimensional fixed point problems, namely, -$1$,$0$, and $+1$, only. 

\begin{definition} The game-equilibrium pair $(G, \s^*)$ is \textit{monogenic} if $\s^*$ is an isolated completely mixed equilibrium of $G$ and $Jp(0)$ has rank $\geq \kappa-1$. \end{definition}

Fixing a pure strategy set $S = \times_{n}S_n$, the space of payoffs for which every isolated completely mixed equilibrium induces a monogenic pair encompasses cases of isolated completely mixed equilibria where the Jacobian matrix of $p$ might lose full rank by at most one at $0$. In general, it can therefore be a strictly larger class than the usually considered generic class of games where the Jacobian at any completely mixed equilibrium is non-zero (where, therefore, the index of any completely mixed equilibrium is $+1$ or $-1$). 

Theorem \ref{mainthm1} shows that when the Jacobian loses full rank by at most one, only one new possibility of index arises for isolated completely mixed equilibria relative to the generic case, namely, the index $0$ (this is result $(a)$ in Theorem \ref{mainthm1}). We also show in item $(b)$ of the same theorem that in the monogenic case, non-zero index fully characterizes payoff-robustness, sharpening the characterization of non-zero index components in \cite{PP2025} and \cite{GW2005} on which we commented in the introduction. In section \ref{SEC4} we comment on how this result holds even in cases where $\s^*$ is not completely mixed.

\begin{theorem}\label{mainthm1}  Suppose $(G, \s^*)$ is monogenic. The following statements hold: 
\begin{enumerate} 
\item[(a)]  The index of $\s^*$ is either $0, +1$ or $-1$;

\item[(b)] $\s^*$ has index $0$ if and only if $\s^*$ is not payoff-robust.   
\end{enumerate} 
\end{theorem}

\begin{proof}We start with the proof of item (a) of the theorem. For that we prove an auxiliary lemma. Let $I = \langle p_1,...,p_\kappa \rangle$ be the ideal generated by the polynomial map $p$ associated to $\s^*$. 

\begin{lemma}\label{lemmaux} The pair $(G, \s^*)$ is monogenic if and only if there exists a non-zero real-analytic function $g(w)$ with $g(0) = 0$ such that $\mathbb{R}[[X_1,....,X_\kappa]]/I  \cong_{\Re} \mathbb{R}[[w]]/ \langle g(w) \rangle$ as $\mathbb{R}$-algebras.\end{lemma}

Notice that, provided the lemma holds, statement (a) is almost immediate: given the $\mathbb{R}$-algebra isomorphism expressed in the lemma and $\mathbb{R}[[w]]/ \langle g(w) \rangle \cong_{\Re} \Re[w]/ \langle w^m \rangle$, where $m$ is the order of $g$, it follows immediately that $\mathbb{R}[[X_1,....,X_\kappa]]/I$ is finite dimensional. By Proposition \ref{auxprop} we then have that $| \text{deg}_0(p)| = \text{dim}_{\mathbb{R}}\mathbb{R}[w]/\langle w^m \rangle - 2 \text{dim}_{\mathbb{R}} I^w_{max} = | \text{deg}_0(g) |$, where $I^w_{max}$ is a maximal square-zero ideal in $\mathbb{R}[w]/ \langle w^m \rangle$ and the last equality holds by Theorem 1.1 of \cite{EL1977} (since $g$ represents a smooth function germ at $0$). Now, it is known (see for example \citet{BO1953}) that in the case of one-dimensional real maps the local degrees can only be $0$, $+1$ or $-1$, which concludes the proof of item (a).

\begin{proof}[Proof of Lemma \ref{lemmaux}]Suppose $(G, \s^*)$ is monogenic. After a change of variables, without loss of generality we can assume that the submatrix of the Jacobian of $p$ at $0$ given by $\frac{\partial (p_1,...,p_{\kappa-1})}{\partial (x_1,...,x_{\kappa-1})}(0)$ is invertible. Define $\Psi : \mathbb{R}^{\kappa} \to \mathbb{R}^{\kappa}$ as $\Psi(x) = (p_1(x),...,p_{\kappa-1}(x), x_{\kappa})$. The Jacobian of $\Psi$ at $0$ is then invertible. By the inverse function theorem, $\Psi^{-1}|_{U} = \Phi: U \to V$ is an analytic local diffeomorphism over neighborhoods $U$ and $V$ of $0$. Let now $u = (u_1,...,u_{\kappa-1})$ and define on $U$ the map $H(u,w) = p(\Phi(u,w)) = (u, h(u,w))$, where $h$ is analytic by construction. The analytic diffeomorphism $\Phi$ now induces an isomorphism $\mathbb{R}[[X_1,....,X_\kappa]]/\langle p_1,...,p_\kappa \rangle \cong_{\Re} \mathbb{R}[[u,w]]/\langle u, h(u,w) \rangle$ (one defines the pullback isomorphism $\Phi^{\#}: \mathbb{R}[[X_1,...,X_{\kappa}]] \to \mathbb{R}[[u,w]]$ which takes $X_i \mapsto \Phi_i(u,w)$ and for each $F \in \mathbb{R}[[X_1,...,X_{\kappa}]]$, $\Phi^{\#}(F) = F(\Phi(u,w))$. Computing the pullback on a generator $p_i$  gives $\Phi^{\#}(p_i) = p_i(\Phi(u,w))$, which implies $\Phi^{\#}(I) = \langle u, h(u,w) \rangle$. Passing then the isomorphism to quotients gives the result. Note that the isomorphism constructed is also $\mathbb{R}$-linear, so it is a $\mathbb{R}$-algebra isomorphism).

Each class of $\mathbb{R}[[u,w]]/\langle u, h(u,w) \rangle$ can be represented by a power series in $w$ only, since each variable in $u$ is zeroed by the ideal $\langle u, h(u,w) \rangle$. This implies that we have an $\mathbb{R}$-algebra isomorphism $\mathbb{R}[[w]]/\langle h(0,w) \rangle \cong_{\Re} \mathbb{R}[[u,w]]/\langle u, h(u,w) \rangle$. Therefore we have obtained that $$\mathbb{R}[[X_1,....,X_\kappa]]/\langle p_1,...,p_\kappa \rangle \cong_{\Re} \mathbb{R}[[w]]/ \langle g \rangle,$$ where $g = h(0, \cdot)$ as $\mathbb{R}$-algebras. Note now that $g$ must be non-zero (if otherwise, $0$ would not be an isolated root of $p$). This concludes the proof of the ``only if'' part of Lemma \ref{lemmaux}.

Now we prove the sufficiency of the Lemma \ref{lemmaux}. Suppose then $A = \mathbb{R}[[X_1,....,X_\kappa]]/\langle p_1,...,p_\kappa \rangle \cong_{\Re} \mathbb{R}[[w]]/\langle g(w) \rangle = Q$. The algebra $Q$ has a unique maximal ideal given by $\mathfrak{M} = \langle w \rangle / \langle g(w) \rangle$. The set $\mathfrak{M}/\mathfrak{M}^2$ can be viewed as real vector space obtained as the quotient of the subspace $\mathfrak{M}$ by the subspace $\mathfrak{M}^2 \subset \mathfrak{M}$. Its dimension is $0$ or $1$.

The local ring $\mathbb{R}[[X_1,..,X_{\kappa}]]$ has the unique maximal ideal $\mathfrak{m} = \langle X_1,...,X_{\kappa} \rangle$ and $A$ has a unique maximal ideal denoted $\mathfrak{m}_A = \mathfrak{m}/I$, with $I = \langle p_1,...,p_{\kappa} \rangle$. Similarly to the above $\mathfrak{m}/\mathfrak{m}^2$ is a real vector space and the isomorphism between $A$ and $Q$ implies that $\mathfrak{m}_A/\mathfrak{m}^2_A \cong \mathfrak{M}/\mathfrak{M}^2$ as real vector spaces. 

We can write $\mathfrak{m}^2_A = \frac{\mathfrak{m}^2+I}{I}$. Therefore, 

$$\mathfrak{m}_A/\mathfrak{m}^2_A \cong \frac{\mathfrak{m}/I}{(\mathfrak{m}^2 + I)/I} \cong \frac{\mathfrak{m}}{\mathfrak{m}^2 + I},$$ 

as real vector spaces. Let $q: \mathfrak{m} \to \frac{\mathfrak{m}}{\mathfrak{m}^2}$ be a linear quotient map associating a power series to its linear part.  Consider $\ell_j = (\frac{\partial p_j}{\partial x_1}(0),...,\frac{\partial p_j}{\partial x_{\kappa}}(0))$. Then $q(I) =\text{span}_{\mathbb{R}}\{\ell_1,...,\ell_\kappa\}$. Therefore, we have the following isomorphisms of real vector spaces 
$$ \frac{\mathfrak{m}}{\mathfrak{m}^2 + I} \cong \frac{\frac{\mathfrak{m}}{\mathfrak{m}^2}}{q(I)} \cong \frac{\mathbb{R}^{\kappa}}{\text{span}_{\mathbb{R}}\{\ell_1,...,\ell_\kappa\}}.$$

The dimension of $\text{span}_{\mathbb{R}}\{\ell_1,...,\ell_\kappa\}$ equals the rank of $Jp(0)$. Therefore, the dimension of $\frac{\mathbb{R}^{\kappa}}{\text{span}_{\mathbb{R}}\{\ell_1,...,\ell_\kappa\}}$ is $$ \kappa - \text{rank}Jp(0) = \text{dim}_{\mathbb{R}}(\frac{\mathfrak{m}_A}{\mathfrak{m}^2_A}).$$ Therefore, $\text{rank}Jp(0)$ is either $\kappa-1$ or $\kappa$, showing that $(G, \s^*)$ is monogenic. This concludes the proof of Lemma \ref{lemmaux}. \end{proof} 

\begin{remark} Observe that the sufficiency direction only requires $g$ to be a
non-zero formal power series; the ``only if'' direction shows that
$g$ can in fact be chosen to be real analytic. \end{remark}

We now prove $(b)$. Since it is known that the absence of payoff-robustness implies $0$ index, we focus on showing the opposite direction. We say a map $p: \Re^m \to \Re^n$ with $p(a) = 0$ is \textit{locally surjective at 0 (around $a$)} if for any neighborhood $U \subset \Re^m$ of $a$, there exists a neighborhood $V \subset \Re^n$ of $0$ such that $V \subset p(U)$. This will require the following auxiliary lemma:

\begin{lemma}\label{lem2}Suppose that the pair $(G, \s^*)$ is monogenic and let $g$ be the real analytic function obtained in Lemma \ref{lemmaux}. If $\s^*$ has index zero then the order of $g$ (about 0) is even. Conversely, if the order of $g$ is even, then $p$ is not locally surjective at zero (and therefore $\s^*$ has index zero) . \end{lemma} 

\begin{proof}[Proof of Lemma \ref{lem2}] Suppose first that $\s^*$ has index zero. Proposition \ref{auxprop} then implies that $\text{dim}_{\mathbb{R}}\mathbb{R}[[w]]/\langle g(w) \rangle$ is even. Fixing the local order in $\mathbb{R}[[w]]$, a basis for $\mathbb{R}[[w]]/\langle g(w) \rangle$ comprises those monomials $w^{\alpha}$ which do not belong to the leading term ideal $\langle LT[\langle g(w) \rangle] \rangle$. Since this leading term ideal is generated by the term corresponding to ord($g$), say, $cw^{k}$, for some $c \neq 0$,  it follows that a basis for $\mathbb{R}[[w]]/\langle g(w) \rangle$ is given by $\{1,....,w^{k-1}\}$, which then implies that $k$ is an even positive integer. 

For the opposite direction, suppose that $g$'s order is even.  Let $cw^{k}$ be the term corresponding to the order of $\text{ord}(g) = \text{ord}(h(0,w))$ and suppose without loss of generality that $c>0$ (the proof for $c<0$ is entirely analogous). For a sequence $\d_j \uparrow 0$, let $q_{j} = (0, \d_j) \in \mathbb{R}^{\kappa-1} \times \mathbb{R}$. Then $p(\Phi(u,w)) = (0, \d_j)$ implies that $h(0,w) = \delta_j < 0$. Take $U \subset \mathbb{R}^{\kappa-1} \times \mathbb{R}$ a neighborhood of $0$ such that $h(0,w) >0$, for all $(0,w) \in U \setminus \{(0,0)\}$. Such neighborhood exists because the order of $h(0,w)$ is even. Hence, there exists $j_0$ such that for each $j \geq j_0$, $(p \circ \Phi) ^{-1}(q_j) \cap U = \emptyset$. Therefore, $p \circ \Phi$ is not locally surjective at $0$ and since $\Phi$ is a local diffeomorphism, $p$ is not locally surjective at $0$, as well. This, together with $0$ being an isolated root of $p$, implies that $\text{deg}_0(p) = 0$. Proposition \ref{degindexequal} implies then that the index of $\s^*$ equals $0$. This concludes the proof of the lemma. \end{proof}

We now conclude the proof of part $(b)$.  Let $\hat{J}$ be the ideal defined by the following set of equations in $Z_n = (Z_{s^0_n},...,Z_{s^{k_n}_n})$ and $Y_n = (Y_{s^0_n},....,Y_{s^{k_n}_n})$ (below, $\mathbbm{1}$ denotes the vector with $1$ in all coordinates),

$$ Z_n - Y_n - G^{n}(Y) = 0, \forall n \in \mathcal{N}$$   $$ (Z_n - Y_n) \cdot (e_{s_n}  - Y_n) = 0, \forall s_n \in S_n, n \in \mathcal{N}$$ $$ Y_n \cdot \mathbbm{1}  -1 = 0, n \in \mathcal{N}$$

Let $\hat{z}$ be the solution to $(\text{proj}\circ \theta^{-1}_G)(z) = 0$, for which $r_n(\hat z_n) = \s^*_n$ for all $n \in \mathcal{N}$. Let us change the variable $Y_n$ to $Y_n + \s^*_n$ and $Z_n$ to $Z_n + \hat{z}_n$ so that the isolated completely mixed solution $Y_n = r_n(\hat z_n)$ and $Z_n = \hat{z}_n$ is translated to zero. Under this change of variables, the ideal generated by the set of equations is denoted by $J$. Now, since $(\text{proj}\circ \theta^{-1}_G)(z)$ is a polynomial map in a neighborhood of $\hat z$ and considering the change of variables from $Z_n$ to $Z_n + \hat{z}_n$, denote by $T(Z)$ this polynomial map after the change of variables. The proof of part $(b)$ will follow from verifying that the sequence of $\mathbb{R}$-algebra isomorphisms holds:    

$$ \frac{\mathbb{R}[[X_1,...,X_{\kappa}]]}{\langle p_1,...,p_{\kappa} \rangle} \cong_{\Re} \frac{\mathbb{R}[[Z,Y]]}{J} \cong_{\Re} \frac{\mathbb{R}[[Z]]}{\langle T(Z) \rangle}, $$

where $\langle T(Z) \rangle $ is the ideal generated by the polynomial coordinate maps of $T(Z)$. Assume first that the sequence of isomorphisms above holds in order to conclude the proof. Because $(G,\s^*$) is monogenic and $\s^*$ has index $0$, the Jacobian $Jp(0)$ must lose full rank by exactly $1$. The chain of isomorphisms above then implies that the Jacobian of $T$ at $0$ loses rank by exactly $1$ as well (the reasoning for this is the exact same employed in the sufficiency proof of Lemma \ref{lemmaux}). Applying the same reasoning we applied to $p$ in the proof of Lemma \ref{lem2}, we obtain that $T$ is not locally surjective at $0$. Therefore,  $\text{proj} \circ \theta^{-1}_G$ is not locally surjective at  $0$, which is equivalent to $\s^*$ not being payoff-robust, which implies our desired result.

We now construct the claimed sequence of isomorphisms: Consider $\varphi: \mathbb{R}[[Z]] \to \mathbb{R}[[Z,Y]]/J$ taking $Z \mapsto [Z]$. Note that the relations in $J$ determine that $Y_n + \s^*_n$ is the retraction of $Z_n + \hat{z}_n$ to the affine space generated by player $n$'s strategy, when $(Z_n,Y_n)$ takes values close to zero. As this retraction can be expressed as a polynomial in the variables of $Z_n$, it implies that each $Y_n$ can be substituted by a polynomial $q_n(Z_n)$ in the variables $Z_n$ in the algebra $\mathbb{R}[[Z,Y]]/J$. Therefore, each equivalence class has a representative in $\mathbb{R}[[Z]]$. Therefore, $\varphi$ is surjective. Now, the kernel of $\varphi$ is simple to compute: note that $(\text{proj} \circ \theta^{-1}_G)_n(z) = z_n - r_n(z_n) - G^n(r_{-n}(z_{-n}))$. Therefore, substituting $Y_n = q_n(Z_n)$ in $J$, the ideal $J$ reduces to $\langle T(Z) \rangle$. Therefore, the kernel of $\varphi$ is $\langle T(Z) \rangle$, which proves the second isomorphism, by passing $\varphi$ to the quotient. Note that the isomorphism is $\mathbb{R}$-linear. 

For the first isomorphism, the first set of relations in $J$ permits us 
to eliminate $Z$ by making the indicated substitution for $Z_n$. After 
making this substitution in the remaining relations and using the 
normalization equations to eliminate one strategy coordinate for each 
player, the resulting quotient algebra is precisely
\[
\frac{\mathbb{R}[[X]]}{\langle p_1,\ldots,p_\kappa\rangle}.
\]
The implied isomorphism is also $\mathbb{R}$-linear. This concludes the proof of the sequence of isomorphisms and the proof of the theorem.\end{proof} 

\begin{remark} As can be seen from the proof of item $(a)$ of Theorem \ref{mainthm1}, this makes no use of the particular fact that the polynomials  originate from equilibrium conditions in a finite game. As long as $0$ is isolated as a root of a polynomial map $p$, it holds.  This is not the case with item $(b)$, which makes use of the particular polynomials describing completely mixed equilibria of the finite game. \end{remark}

\begin{corollary}\label{Cor} Suppose $(G, \s^*)$ is monogenic. The index of $\s^*$ is zero if and only if the dimension of the $\mathbb{R}$-algebra $\mathbb{R}[[X_1,....,X_{\kappa}]]/ \langle p_1,...,p_{\kappa} \rangle$ is even.\end{corollary} 

\begin{proof}Follows from the formula in Proposition \ref{auxprop} and item (a) of Theorem \ref{mainthm1}. \end{proof} 

\begin{example}\label{example} As an illustration of the results of Theorem \ref{mainthm1}, we revisit the example in subsection \ref{example2}. As observed in Corollary \ref{Cor}, because the dimension of the algebra $Q$ is even, it follows that the index of $\s^*$ is $0$. From part $(b)$ of Theorem \ref{mainthm1}, it follows that $\s^*$ is not payoff-robust. In p.151 of \citet{AM2018}, a specific payoff perturbation is constructed to show that $\s^*$ is not robust. Theorem \ref{mainthm1} allows us to bypass this construction and simply check the dimension of $Q$.\end{example}

A natural question is whether the full-dimension of the Jacobian of $p$ at $0$ is equivalent to nonzero index, at least in the monogenic case. This is not true: it is possible that a completely mixed and isolated equilibrium of a game has an associated Jacobian without full rank and yet has index different than $0$, as Proposition \ref{falseconjecture} shows. This is a consequence of the proof of Theorem $7$ in \citet{RD2003}. 


\begin{proposition}\label{falseconjecture}
There exists a finite game with three players where one player has $2$ strategies and the other two  have $3$ strategies each, with a completely mixed equilibrium $\s^*$ such that $(G, \s^*)$ is monogenic with rank of $Jp(0)$ equal to $\kappa-1$ and the index is nonzero.\end{proposition} 

In order to make the proof of Proposition \ref{falseconjecture} more self-contained we make a few observations on the method used by Datta (ibid.). The proposition relies on a simple way of writing a one-variable polynomial system as a logically equivalent system of polynomial equations that ultimately describe the completely mixed equilibrium of a game. If $q(x) = \alpha_dx^d + \alpha_{d-1}x^{d -1} + ... +\a_0 = 0$ is a one variable polynomial equation, then it can be rewritten iteratively as $q(x) = (...(\alpha_d x + \alpha_{d-1})x + \a_{d-2})x + ... +\a_1)x + \a_0 =0$. We can then ``disaggregate'' this polynomial equation  into a system of multiaffine equations in two variables $c_i's$ and $x$: for example, one could write: 

\begin{equation}
  \begin{aligned}
    0 &= \alpha_0 + x(s_1c_1 + \d_1)\\
      s_1c_1 + \d_1 &= \a_1 + x(s_2c_2 + \d_2)\\
		&\vdots \\
       s_{d-1}c_{d-1} + \d_{d-1} &= \a_{d-1} + x\a_d 
  \end{aligned}
\end{equation}

Note that each such equation is multiaffine in the variables $c_i$ and $x$, and $s_{i}$ and $\d_{i}$ are constants chosen so that the solutions in $c_i$ are probabilities (supposing that $q$ has a solution  $x \in (0,1)$). Lemma 4 in Datta (ibid.) shows that systems such as the one above correspond to the polynomial equations describing the completely mixed equilibria of a game, namely, the payoff indifference equations of a third player (one that does not have the strategies $c_i$ or $x$).

\begin{proof}[Proof of Proposition \ref{falseconjecture}] The proof is a corollary of the proof of Theorem 7 in \citet{RD2003}, so we simply sketch the reasoning for completeness, using the notation of that theorem. Consider a real polynomial in one-variable $q(a)$ such that $q(0) = 0$, the degree of $q$ is $4$ and the order of $q$ is $3$. Fix some $a^* \in (0,1)$ and consider then the translated polynomial $p(a) = q(a - a^*)$. The polynomial $p$ has an isolated root at $a^*$. By Theorem 7 in \citet{RD2003},  there exists a three-player (Alice, Bob and Critter) game $G$ where player Alice has 2 pure strategies, and Bob and Critter have $3$ strategies each. Consider the polynomial system $Q(a, b_1, b_2, c_1,c_2) =0$ defining the completely mixed equilibria of the game as outlined in p. 430, ibid. (we follow the notation used in Datta's proof, where $a$ denotes one of the mixed strategy coordinates of Alice, $b_1$, $b_2$ denote two of the mixed strategy coordinates of Bob, and $c_1, c_2$ two mixed strategy coordinates of Critter). By construction of the system, there exists a uniquely determined completely mixed equilibrium $\s^* = (a^*, b^*_1, b^*_2, c^*_1, c^*_2)$ (here we omitted a strategy variable for each player, since it is uniquely determined by probabilities) by the root $a^*$. Now, change variables in the system $Q=0$, so that the completely mixed equilibrium $\s^*$ is translated to $0$ in a new system $F(A,B_1,B_2,C_1,C_2) = 0$ with an isolated root at $0$. The system can now be reduced, through a series of substitutions and eliminations of variables (see p. 427, ibid. for an explanation), to $q(a)=0$, thus implying that $\mathbb{R}[[A,B_1,B_2,C_1,C_2]]/ \langle F \rangle \cong_{\Re} \mathbb{R}[[a]]/ \langle q(a) \rangle$, as $\mathbb{R}$-algebras. Therefore $(G,\s^*)$ is monogenic (by Lemma \ref{lemmaux}). Moreover, the Jacobian $JF(0)$ loses rank by $1$, since the maximal ideal of $\mathbb{R}[[a]]/\langle q(a) \rangle$ is non-zero (the elements of the ideal $\langle q(a) \rangle$ all have order at least $3$, so the maximal ideal $\mathfrak{M} = \langle a \rangle / \langle q(a) \rangle $ is non-zero. This in turn implies that $\mathfrak{M}/\mathfrak{M}^2$ has dimension $1$, which determines that $JF(0)$ does not have full rank). Finally, recall that the dimension of the $\mathbb{R}$-algebra $\mathbb{R}[[x]]/\langle q(x) \rangle$ is odd, since ord$(q) =3$, which implies that the index is nonzero (by Lemma \ref{lem2}).\end{proof}

\begin{remark}Note that, since the proof of Theorem $7$ in \citet{RD2003} is constructive, one could obtain the game of Proposition \ref{falseconjecture} explicitly. \end{remark}  

Generalizing from the analysis in the last example, we have the more general statement about three-player games with two strategies per player.  

\begin{proposition}Let $G$ be a three-player game with two strategies per player. Let $\s^*$ be a completely mixed equilibrium of $G$ that is isolated. Then $(G,\s^*)$ is monogenic. Therefore, the indices of isolated equilibria in such games can only be $0$, $-1$ and $+1$. \end{proposition}

\begin{proof} Let $p = p_1 \times p_2 \times p_3$ be the polynomial map associated to $(G,\s^*)$ in variables $x$ (player $1$'s strategy), $y$ (player $2$'s strategy) and $z$ (player $3$'s strategy). We view $p_1$ as the polynomial of player $1$, that is, the polynomial featuring variables corresponding to strategies of players $2$ and $3$ only. Analogously, we view $p_2$ as the polynomial of player $2$ and $p_3$ as the polynomial of player $3$. Because $0$ is an isolated root of $p$ by assumption, the linear parts of $p_1, p_2$, and $p_3$ must be such that they feature linear monomials in $x, y,$ and $z$: if not, by the multilinearity of each $p_i$, the absence of such a monomial in the linear part of $p_i$ for each $i$ implies that $0$ is not isolated. Now, in each $p_i$, there is at least one variable absent, namely, the variable corresponding to the strategy of player $i$. Let $\ell(p_1)$, $\ell(p_2)$ and $\ell(p_3)$ be the linear parts of polynomials $p_1, p_2$, and $p_3$, respectively. Therefore, we have the following possibilities: (1) each $\ell(p_i)$ features at least one variable; (2) exactly one $\ell(p_i) =0$, with the others $\ell(p_j), j \neq i$ featuring all the remaining variables.  

In the first case, suppose without loss of generality that $\ell(p_2)$  has a linear term in $x$. Then $y$ or $z$ must feature in $\ell(p_1)$. Without loss of generality, suppose that it is $y$. Fix the degree anticompatible lex order $>$, with $x > y > z$. Construct the $3 \times 3$ matrix $M$ as follows: entry $M_{ij}$ is the coefficient of the variable of lex-order $j$ in equation $i$. As the linear term in $x$ cannot feature in $\ell(p_1)$, $M_{11} =0$. Similarly, $M_{22} = M_{33} = 0$. Since $y$ features in $\ell(p_1)$, then $M_{12} \neq 0$. We then have that the rank of $M$ is at least $2$ (for example, by Gaussian elimination). This implies that  $(G, \s^*)$ is monogenic. 

In the second case, suppose without loss of generality that 
$\ell(p_1)=0$. Since $z$ can feature only in $\ell(p_1)$ or 
$\ell(p_2)$, it must feature in $\ell(p_2)$. Similarly, $y$ can 
feature only in $\ell(p_1)$ or $\ell(p_3)$, so it must feature in 
$\ell(p_3)$. Fixing the degree anticompatible lex order $>$, with 
$x>z>y$, and constructing $M$ as above, implies that $M$ has rank 
at least $2$, which implies that $(G,\s^*)$ is monogenic.\end{proof}

A natural question to ask now is what happens with the indices of isolated and completely mixed equilibria outside the class of monogenic games. The answer to this question is given in Theorem \ref{finalthm}. This theorem can be viewed as a strengthening of the result in Proposition 3.1 in \cite{GSS2004}, which proved that for any integer $n$, one can obtain a (bi-matrix) game with a nondegenerate component of equilibria of index $n$. Theorem \ref{finalthm} can also be viewed as the game-theoretic analog of the known fact of fixed point theory: for any integer $n$, it is possible to show that there exists a map from $\Re^2$ to $\Re^2$ with an isolated fixed point at $0$ with index $n$. The theorem therefore shows that, in particular, indices of isolated equilibria are as rich as the indices of isolated fixed points.

\begin{theorem}\label{finalthm}Given any integer $n \in \mathbb{Z}$, there exists a three-player game $G$ with an isolated completely mixed equilibrium with index $n$.\end{theorem}

\begin{proof} 
We first construct polynomials with an isolated root at $0$ with any fixed degree. Suppose $m>0$. Consider the complex function $z = (x+ iy) \mapsto (x - iy)^{m}$. Let $F: \mathbb{R}^2 \to \mathbb{R}^2$, be the realization of the complex function as a real polynomial map. The polynomial map $F$ has a unique root at $0$. By construction the winding number of $F$ around $0$ is $-m$, thus implying that  $\text{deg}_0(F) = - m$. Alternatively, consider $z \mapsto z^m$ and the realization $F$ of this polynomial to obtain the unique root at $0$ with topological degree $m$. Now suppose $m =0$. Consider the real polynomial map $F(x,y) = (x^2, y)$. This polynomial has an isolated root at $0$. We claim the topological degree of $0$ with respect to the map $F$ is $0$: first we note that the dimension of the $\mathbb{R}$-algebra $\mathbb{R}[[X,Y]]/\langle X^2, Y \rangle$ is $2$, since $\{1,X\}$ are the only surviving monomials under the relations imposed by the ideal (and $\{1,X\}$ is a basis for the space). Therefore, the formula for the computation of the absolute value of the topological degree in Proposition \ref{auxprop} applies. Now note that the maximal square-zero ideal of $\mathbb{R}[[X,Y]]/\langle X^2, Y \rangle$ is generated by $X$. It is, therefore, one-dimensional. The formula then implies that $0$ has topological degree $0$. 

In order to apply Theorem 5 in \citet{RD2003} and obtain a finite game associated to the zero-set of each of the polynomials above, we change variables so that the unique root $0$ is translated to the point where all coordinates are equal to $1/3$, satisfying the requirement of that theorem. Of course, the topological degree of $0$ in either case is invariant under this translation. Applying Theorem 5, we obtain the required game. Now, one needs to verify that the topological degrees of the equilibria are maintained under the new system describing the completely mixed equilibria. For that, consider $F$ as one of the polynomial maps constructed above. We checked that in the case of $F(x,y) = (x^2,y)$, the algebra $\mathbb{R}[[X,Y]]/\langle X^2, Y \rangle$ is finite dimensional, and a similar argument yields the same finite-dimensionality for $\mathbb{R}[[X,Y]]/\langle F_1, F_2\rangle$ in the case $F$ is specified as the other polynomials. Applying Theorem 5 in Datta and Proposition \ref{auxprop}, this implies that the absolute value of the topological degree of the equilibrium is maintained. However, for the non-zero degree case, we need to show that the signs are maintained as well. For that purpose, we use  Theorem 1.2 in \cite{EL1977}, which provides a formula for the topological degree based on the signature of the bilinear form defined over the $\Re$-algebra. Proposition 2 in \cite{EL1978} yields that the minimal non-zero ideal in $\mathbb{R}[[X,Y]]/ \langle F_1, F_2 \rangle$ is one-dimensional and is generated by the class of the Jacobian determinant of this polynomial map. Let  $\varphi: \mathbb{R}[[X,Y]]/ \langle F_1, F_2 \rangle \to  \mathbb{R}[[X_1,...,X_{\kappa}]]/ \langle p_1,..., p_{\kappa} \rangle$ be an isomorphism of $\Re$-algebras obtained, as before, by the successive reductions of $\mathbb{R}[[X_1,...,X_{\kappa}]]/ \langle p_1,..., p_{\kappa} \rangle$ to $\mathbb{R}[[X,Y]]/ \langle F_1, F_2 \rangle$, via the substitutions defined in the proof of Theorem 5 (where $\mathbb{R}[[X_1,...,X_{\kappa}]]/ \langle p_1,..., p_{\kappa} \rangle$ is the $\mathbb{R}$-algebra associated to the system of completely mixed equilibria obtained from $F$ by applying Theorem 5 in Datta).
The  minimal non-zero ideal of each $\Re$-algebra maps under $\varphi$ to the minimal non-zero ideal of the other. By rescaling one payoff-indifference polynomial, as permitted by
Lemma 4 in \citet{RD2003}, we may assume that $\varphi$ maps the class of the Jacobian determinant of $F$ to the class of the Jacobian determinant of $p$. This implies that if $\phi: \mathbb{R}[[X_1,...,X_{\kappa}]]/ \langle p_1,..., p_{\kappa} \rangle \to \mathbb{R}$ is a linear functional that takes value $1$ in the class of the Jacobian determinant, then the signature of the bilinear form $[a,b] := \phi(ab), a,b \in \mathbb{R}[[X_1,...,X_{\kappa}]]/ \langle p_1,..., p_{\kappa} \rangle$ is equal to the topological degree of $0$ under $p$. The isomorphism $\varphi$ allows us to define the linear functional $\psi = \phi \circ \varphi$ on $\mathbb{R}[[X,Y]]/ \langle F_1, F_2 \rangle$ that takes value $1$ on the class of the Jacobian determinant of $\mathbb{R}[[X,Y]]/ \langle F_1, F_2 \rangle$. Letting $(\cdot, \cdot)$ denote the bilinear form on $\mathbb{R}[[X,Y]]/ \langle F_1, F_2 \rangle $ defined from $\psi$, note that the signature of $[\cdot, \cdot]$ equals the signature of $(\cdot, \cdot)$, by construction. By Theorem 1.2 in Eisenbud and Levine (ibid.), the signature of $(\cdot, \cdot)$ is equal to the topological degree of $0$ under $F$.  The proof is then concluded by observing that the topological degree and index of the equilibrium point agree, by Proposition \ref{degindexequal}. \end{proof}


\section{Extensions}\label{SEC4}
There are straightforward extensions of the index formula presented in Proposition \ref{auxprop}. First, the index of equilibria is invariant under the elimination of duplicate strategies (i.e., mixed strategies of a player that are added to the game as pure strategies of that player).  This implies that if there is a non-degenerate component of equilibria of a game that reduces to a singleton completely mixed equilibrium once duplicate strategies are eliminated, then this implies that the indices of the component and the equivalent singleton are the same. As an example, if we add to the three-player game in Example \ref{example} the completely mixed equilibrium mixed strategy of player 1 as a pure strategy, then a nondegenerate component of equilibria arises, and this component has the same index as the singleton once the duplicate strategy is eliminated. Therefore, if the singleton is monogenic in the reduced game, then the results of Theorem \ref{mainthm1} apply as well. 

Second, the index is also invariant under the addition or elimination of strictly inferior replies to the equilibrium: therefore, if a singleton equilibrium lies on the boundary of the strategy set and after elimination of inferior replies, it becomes completely mixed in the reduced game, then the results of this paper apply to such boundary cases. As an illustration, consider the following outside-option game with three players: player 1 first decides between the outside-option (where the game ends and players get $(-1,2,7)$) or coming In and playing the simultaneous move game defined in Example \ref{example}. Notice that the payoff to player 1 of the equilibrium where player 1 chooses In and then the completely mixed equilibrium ($\s^*$) is played is $0$ (this equilibrium is on the boundary of the mixed strategy set). Any strategy of player 1 where $Out$ is chosen first is a strictly inferior reply to that equilibrium. Therefore, the index of the equilibrium In-$\s^*$ is $0$ in the outside-option game. 

A third (less straightforward) extension of the results in this paper concerns extensive form games with $N$ players, namely, equilibrium outcomes that are isolated and completely mixed (i.e., where the distribution over terminal nodes induced by equilibria is completely mixed and isolated in the space of equilibrium distributions).\footnote{For a generic choice of terminal payoffs in a fixed game tree with perfect recall, one has that each one of the finitely many connected components in mixed strategies of the game yields a unique equilibrium outcome distribution.} In order to compute the index of an isolated completely mixed equilibrium outcome, one first needs to compute the strategies inducing the outcomes in \textit{enabling form} (see \cite{GW2002} or \cite{BS1996}): since the outcome is isolated and completely mixed, there is a unique enabling strategy profile $e^*$ inducing it, which is in the interior of the enabling strategy set of the players and is isolated in the space of equilibria in enabling form. In each enabling strategy space $E_n$ of player $n$, consider a simplex $\Delta_n$ satisfying the following: (1) $\D_n$ is contained in the interior of $E_n$; (2) $\D_n$ has the same dimension as $E_n$; (3) $\D_n$ contains $e^*_n$ in its interior; (4) $\D = \times_n \D_n$ contains no other equilibrium of the game besides $e^*$. Consider now the enabling form of this game (see section 4 of \cite{LP2023}): if $V_n: \times_n E_n \to \Re$ is player $n$'s multiaffine payoff function in enabling strategies, let $G_n(v_1,...,v_N) := V_n(v_1,...,v_N)$, where $v_n$ is a vertex of $\D_n$. We can view each vertex  $v_n$ as a pure strategy of player $n$ and therefore define a normal form game $G$ by considering the mixed extension over these pure strategies. This normal form game now has a completely mixed equilibrium $\s^*$ (which corresponds to $e^*$), and the index of $\s^*$ can be computed through the formula presented in Proposition \ref{auxprop}. This index of the component of equilibria in mixed strategies that induces the equilibrium outcome distribution is identical to the normal form index of $\s^*$ in $G$ (see Proposition 3.10 in \cite{LP2023}). This procedure outlines a method to compute the index of an isolated completely mixed equilibrium outcome of an extensive form game by reducing the initial game to a normal form game with an associated isolated completely mixed equilibrium that has the same index. 

\begin{example}\label{extensiveformexample} We illustrate the method described above in the following 3-player extensive form game.
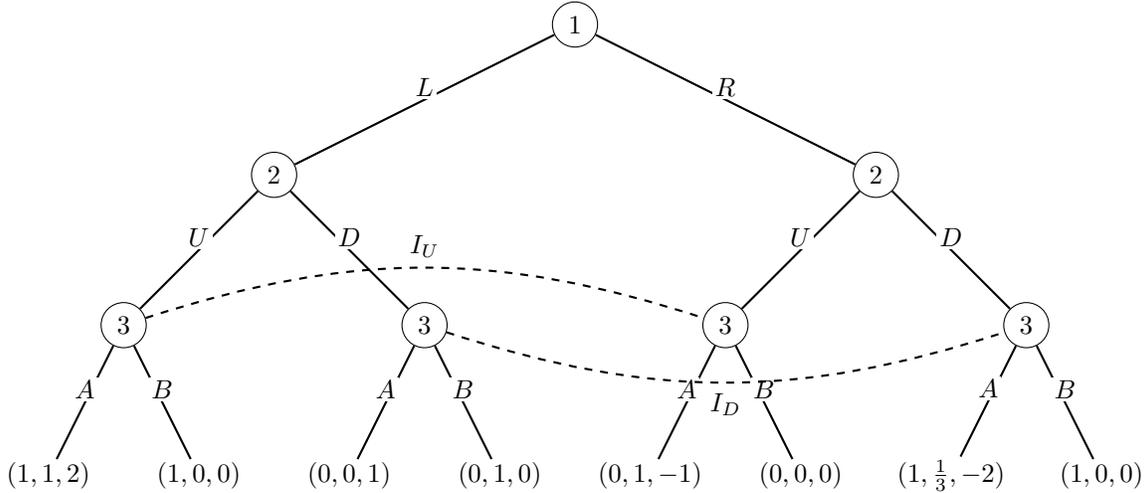
\begin{center}
\begin{tikzpicture}[
  x=1cm,y=1cm,
  every node/.style={font=\small},
  decision/.style={draw,circle,inner sep=1.2pt,minimum size=6mm},
  terminal/.style={inner sep=1pt},
  act/.style={midway,fill=white,inner sep=1pt},
  line/.style={thick}
]
\node[decision] (P1) at (0,0) {1};
\node[decision] (P2L) at (-4,-2) {2};
\node[decision] (P2R) at (4,-2) {2};

\node[decision] (LU) at (-6,-4) {3};
\node[decision] (LD) at (-2,-4) {3};
\node[decision] (RU) at (2,-4) {3};
\node[decision] (RD) at (6,-4) {3};

\node[terminal] (LUA) at (-7,-6) {$(1,1,2)$};
\node[terminal] (LUB) at (-5,-6) {$(1,0,0)$};
\node[terminal] (LDA) at (-3,-6) {$(0,0,1)$};
\node[terminal] (LDB) at (-1,-6) {$(0,1,0)$};
\node[terminal] (RUA) at (1,-6) {$(0,1,-1)$};
\node[terminal] (RUB) at (3,-6) {$(0,0,0)$};
\node[terminal] (RDA) at (5,-6) {$(1,\tfrac13,-2)$};
\node[terminal] (RDB) at (7,-6) {$(1,0,0)$};

\draw[line] (P1) -- (P2L) node[act,above] {$L$};
\draw[line] (P1) -- (P2R) node[act,above] {$R$};

\draw[line] (P2L) -- (LU) node[act,above] {$U$};
\draw[line] (P2L) -- (LD) node[act,above] {$D$};
\draw[line] (P2R) -- (RU) node[act,above] {$U$};
\draw[line] (P2R) -- (RD) node[act,above] {$D$};

\draw[line] (LU) -- (LUA) node[act,above] {$A$};
\draw[line] (LU) -- (LUB) node[act,above] {$B$};
\draw[line] (LD) -- (LDA) node[act,above] {$A$};
\draw[line] (LD) -- (LDB) node[act,above] {$B$};
\draw[line] (RU) -- (RUA) node[act,above] {$A$};
\draw[line] (RU) -- (RUB) node[act,above] {$B$};
\draw[line] (RD) -- (RDA) node[act,above] {$A$};
\draw[line] (RD) -- (RDB) node[act,above] {$B$};

\node[above=6mm] at (P1) {};
\node[left=8mm] at (P2L) {};
\node[left=8mm] at (LU) {};

\draw[dashed,thick] (LU) to[bend left=18] node[midway,above] {$I_U$} (RU);
\draw[dashed,thick] (LD) to[bend right=18] node[midway,below] {$I_D$} (RD);
\end{tikzpicture}
\end{center}

Let $x$ be the probability of $L$, $y$ be the probability of $U$ after $L$, $z$ the probability of $U$ after $R$, $w$ the probability of $A$ in the information set $I_U$ of player 3, and $v$ the probability of $A$ in the information set $I_D$ of player 3. This game has a completely mixed, isolated equilibrium given by: 

$$ x^\ast=\frac12,\qquad y^\ast=\frac13,\qquad z^\ast=\frac23,\qquad w^\ast=\frac14,\qquad v^\ast=\frac34.$$

Because in this game no player can exclude an information set of herself by choosing any of her actions, the set of behavior strategies of each player is equal to her set of enabling strategies. Player 1's enabling strategy set is equal to $E_1 = [0,1]$, Player $2$'s enabling strategy set is equal to $E_2 = [0,1]^2$ (with coordinates $(y,z)$) and Player $3$'s strategy set is equal to $E_3 = [0,1]^2$ (with coordinates $(w,v)$). The multiaffine payoff functions of each player are given by: 

\begin{align*}
V_1(x,y,z,w,v) &= xy + (1-x)(1-z),\\
V_2(x,y,z,w,v) &= x\big(yw+(1-y)(1-v)\big) + (1-x)\big(zw+(1-z)\tfrac{v}{3}\big),\\
V_3(x,y,z,w,v) &= x\big(y(2w)+(1-y)v\big) + (1-x)\big(z(-w)+(1-z)(-2v)\big).
\end{align*}

We define now the auxiliary normal form game $\hat G$ to this enabling form in order to compute the index of that equilibrium. Fix $\varepsilon=0.05$.
Define a 1-simplex for Player 1:
\[
\Delta_1=\mathrm{conv}\{\,x^\ast-\varepsilon,\ x^\ast+\varepsilon\,\}\subset (0,1).
\]
Define 2-simplices for Players 2 and 3:
\begin{align*}
\Delta_2 &= \mathrm{conv}\Big\{(y^\ast-\varepsilon,z^\ast-\varepsilon),\ (y^\ast+\varepsilon,z^\ast-\varepsilon),\ (y^\ast,z^\ast+\varepsilon)\Big\}\subset (0,1)^2,\\
\Delta_3 &= \mathrm{conv}\Big\{(w^\ast-\varepsilon,v^\ast-\varepsilon),\ (w^\ast+\varepsilon,v^\ast-\varepsilon),\ (w^\ast,v^\ast+\varepsilon)\Big\}\subset (0,1)^2.
\end{align*}
$e^\ast = (x^*, y^*, z^*, w^*, v^*)$ lies in the interior of each $\Delta_n$, and $\Delta:=\Delta_1\times\Delta_2\times\Delta_3$ contains no other equilibrium besides $e^\ast$. Let the vertices of $\Delta_j$ be $v_j^1,v_j^2, v_j^3, j=2,3$ as listed above. 

\[
\hat G_n(v_1^i,v_2^j,v_3^k) := V_n(v_1^i,v_2^j,v_3^k).
\]

Viewing vertices as pure strategies of the players in $\hat G$, this defines a normal form game with an isolated completely mixed equilibrium $\s^*$ corresponding to $e^*$ with the same index. The equilibrium $\s^*$ is given by writing $e^*$ in the barycentric coordinates of the simplices $\D_n$: 

\[
x^\ast=\tfrac12 v_1^1+\tfrac12 v_1^2,\qquad
(y^\ast,z^\ast)=\tfrac14 v_2^1+\tfrac14 v_2^2+\tfrac12 v_2^3,\qquad
(w^\ast,v^\ast)=\tfrac14 v_3^1+\tfrac14 v_3^2+\tfrac12 v_3^3.
\]

We now employ the techniques developed in this paper to compute the index of $\s^*$. We first compute the polynomial system given by $p$ defined in the beginning of section \ref{SEC3}. There are five equations in five unknowns $(P,Q_1,Q_2,R,S)$:  

\begin{align}
 p_1 &= 3Q_1+Q_2 =0 \label{eq:p1}\\
 p_2 &= 11R-7S + 28\varepsilon PR + 52\varepsilon PS =0 ,\label{eq:p2}\\
 p_3 &= -7R-13S -44\varepsilon PR + 28\varepsilon PS =0 \label{eq:p3}\\
 p_4 &= -24\varepsilon P Q_1 -24\varepsilon P Q_2 -16P + 6Q_1 + 6Q_2 =0 \label{eq:p4}\\
 p_5 &= -168\varepsilon P Q_1 -72\varepsilon P Q_2 -16P + 18Q_1 + 42Q_2 =0\label{eq:p5}
\end{align}
We compute the topological degree of the root $(P,Q_1,Q_2,R,S)=(0,0,0,0,0)$ using the formula of Proposition \ref{auxprop}, which then gives us the index of the completely mixed equilibrium $\s^*$, which in turn gives us the index of the connected component in mixed strategies inducing the outcome.

First use equation \eqref{eq:p1} to eliminate $Q_2$ and substitute the result in \eqref{eq:p4} and \eqref{eq:p5}, thus obtaining a system in $4$ variables and $4$ equations. Equations   \eqref{eq:p4} and \eqref{eq:p5} now become: 

\begin{align}
&48\varepsilon P Q_1 -16P -12Q_1=0,\label{eq:g4}\\
&48\varepsilon P Q_1 -16P -108Q_1=0.\label{eq:g5}
\end{align}

Fix the degree anticompatible lex order given by  $Q_1 < P < R < S$. Note now that a linear combination of the equations \eqref{eq:g4} and \eqref{eq:g5} (which is of course in the ideal $I$ generated by the polynomial map $p$) gives an equation in one variable only ($Q_1$). Therefore $Q_1 \in \langle LT[I] \rangle$. In addition, from equation \eqref{eq:g4}, we have that $P$ also belongs to $\langle LT[I] \rangle$. From equation $\eqref{eq:p2}$, $S$ also belongs to $\langle LT[I] \rangle$. Therefore, $\mathbb{R}[[Q_1,P,R,S]]/\langle LT[I] \rangle$ has a basis of monomials in $R$ only. But using the combination of $\eqref{eq:g4}$ and $\eqref{eq:g5}$, one gets that $Q_1 = 0$. Substituting this equation in \eqref{eq:g4}, we get $P = 0$ as well. Substituting then the latter in $\eqref{eq:p2}$ and $\eqref{eq:p3}$, we get two linear equations in $R$ and $S$. Eliminating $S$ from these equations, we get a linear equation in $R$, which shows that $R \in \langle LT[I] \rangle$. Therefore, $\mathbb{R}[[Q_1,P,R,S]]/ \langle LT[I] \rangle$ has dimension $1$ with $I_{max} = 0$. Therefore, the absolute value of the index of $\s^*$ is $1$.

\end{example}

\section{Final Remarks}\label{SEC5} There are important cases that do not fit the framework presented in this paper: generic extensive-form games frequently involve non-degenerate components of equilibria that are located on the boundary of the strategy set (in mixed strategies). For example, this occurs in two-player outside-option games in which the outside option is an equilibrium. In such cases, even under the transformations we discussed in the previous paragraph, one might not be able to reduce the equilibrium component to a completely mixed equilibrium after the elimination of duplicates or inferior replies, or recasting the game in enabling form. These examples indicate that more work is needed in order to extend these algebraic methods to non-degenerate components of equilibria which are possibly not completely mixed. 

Finally, there is an important question about the relation between non-zero index and payoff-robustness that this paper was not able to answer: are there completely mixed equilibria that have index $0$, but are payoff-robust? We have demonstrated that this is not possible in monogenic equilibria, but we believe that such an example exists. At first, it seemed to us that this question could be handled by the methods presented in \cite{RD2003} in combination with the algebraic methods developed here: Datta's method allows us to reduce questions about equilibria in games to questions about real algebraic varieties, but this reduction does not seem to preserve payoff-robustness of equilibria: a zero of a polynomial map might be robust to perturbations of the coefficients of the polynomial map, but the corresponding equilibrium to that zero - after applying Datta's methods - might not be payoff-robust. The issue arises because the polynomial systems constructed in Datta that describe the completely mixed equilibria of a game and reduce to the original polynomial map after substitutions are highly non-generic, a fact that yields a payoff space that is possibly higher dimensional than the space of coefficients of the original polynomial map.

\renewcommand{\bibfont}{\small}
\bibliographystyle{abbrvnat}
\bibliography{references}

\end{document}